\documentclass[aps,prx,reprint,twocolumn,superscriptaddress,floatfix,nofootinbib,longbibliography]{revtex4-1}
\usepackage{epsfig,amsmath,amssymb,color,comment,physics}
\usepackage[makeroom]{cancel}
\usepackage[caption=false]{subfig}
\usepackage{float}
\usepackage[countmax]{subfloat}
\usepackage[normalem]{ulem}
\usepackage[english]{babel}
\usepackage{dsfont}
\usepackage{braket}
\usepackage[bookmarks=true,colorlinks,linkcolor=OrangeRed,urlcolor=NavyBlue,citecolor=RoyalBlue]{hyperref}
\usepackage[dvipsnames]{xcolor}
\usepackage{graphicx}
\usepackage{cleveref}
\graphicspath{{./Figures/}}
\usepackage{amsmath, amssymb,latexsym,amsfonts}
\usepackage{orcidlink}
\usetikzlibrary{patterns}
\usepackage{bm}
\usepackage{hyperref}
\usepackage{quantikz}
\usepackage{tikz}

\usepackage{amsthm} 
\theoremstyle{plain}
\newtheorem{theorem}{Theorem}[section]

\theoremstyle{definition}
\newtheorem{definition}[theorem]{Definition}

\usepackage{tabularray}
\usepackage[export]{adjustbox}

\begin{document}
\title{Spectral Fingerprints of Gauge Theories on a Quantum Computer}

\author{Graham Van Goffrier}
\email{gwvg1e23@soton.ac.uk}
\affiliation{School of Physics and Astronomy, University of Southampton, University Road, SO17 1BJ, UK}
\affiliation{Infleqtion UK}

\author{Debasish Banerjee${}^{\orcidlink{0000-0003-0244-4337}}$}
\email{D.Banerjee@soton.ac.uk}
\affiliation{School of Physics and Astronomy, University of Southampton, University Road, SO17 1BJ, UK}

\author{Bipasha Chakraborty}
\email{B.Chakraborty@soton.ac.uk}
\affiliation{School of Physics and Astronomy, University of Southampton, University Road, SO17 1BJ, UK}

\author{Emilie Huffman${}^{\orcidlink{0000-0002-4417-338X}}$}
\email{ehuffman@wfu.edu}
\affiliation{Department of Physics and Center for Functional Materials, 
Wake Forest University, Winston-Salem, North Carolina 27109, USA}

\begin{abstract}
  {\bf Abstract}
  Maximally mixed state spectral sampling is an unbiased quantum algorithm that allows for extraction of a finite-resolution spectral distribution from a Hamiltonian over potentially the entire allowed range of energies. We show how it may be focused on any desired area of the spectrum in order to learn about the full \textit{fingerprint} of the model of interest: from its ground state phenomena such as quantum criticality, obtained from the lowest lying energies, to its thermalization behavior, obtained from the mid-spectrum. We demonstrate this technique specifically on a $(1+1)d$ non-Abelian $SO(3)$ gauge theory, providing a comprehensive analysis of the steps necessary for performing this algorithm, as well as what is possible in the near-term with superconducting quantum hardware, performing simulations with circuits that are $78$ two-qubit gates deep. We show how this algorithm is able to take advantage of emerging dynamical circuit capabilities in near-term hardware to roughly halve the number required qubits, as well as how quantum readout error mitigation is trivial for this method. Along the way, we propose a novel strategy for compiling the controlled-time evolutions needed for spectral sampling by means of Pauli-frame optimizations. We illustrate two physical applications of quantum spectral sampling -- disordered many-body transitions, and mid-spectrum densities of states -- and what postprocessing steps they require beyond the Fourier outputs of the algorithm.

\end{abstract}

\date{\today}
\maketitle

\tableofcontents
\section{Introduction}
 The simulation of quantum many-body systems is one of the most natural applications 
of a quantum computer. Exploratory and advanced efforts in condensed-matter, nuclear 
and high-energy physics have focused on preparing ground states
\cite{Bauer2016,Dumitrescu2018,Motta2019,Farrell2024,Peruzzo2013,Maiti2024,Stetcu2021,Chakraborty2020,
Ciavarella2021,Schuster2023,Lumia2021}, determining a small number of low-lying excitations 
\cite{McClean2016,Higgott2018,Nakanishi2019,Atas2021,Guo2024,Rosanowski:2025fhr}, and following the real-time 
evolution of physically motivated initial states \cite{Smith2019,Roggero2019,Martinez2016,
deJong2021,DAlessio2015,Huffman2021,Nagano2023,Hall2021,Ciavarella2020}. These are 
important goals, but they often can only be tailored to obtain at most a few states of the spectrum of the Hamiltonian (typically the lowest lying energetic states).
The complete energy spectrum, or a suitably coarse-grained representation 
of it, contains complementary information which is important for finite-temperature physics, 
such as the equation of state or the process of thermalization. Moreover, for conformal field theories on spherical geometries, it contains their scaling dimensions.\cite{Cardy1985,Zhu2023} Unbiased methods for full spectral energy extraction are therefore desirable in order to obtain the full \textit{fingerprint} of the theory of interest.

 The full energy spectrum of a theory is encoded in its density of states (DOS). For a Hamiltonian $H$ with 
distinct eigenvalues, $E_\alpha$, and degeneracies, $d_\alpha$, the DOS is defined as
\begin{equation} \label{eq:dos}
    \varrho_H(E) = \Tr~\left[\delta(E-H)\right] = \sum_\alpha d_\alpha~\delta(E-E_\alpha).
\end{equation}
Knowledge of $\varrho_H(E)$ gives direct access to the canonical partition function,
\begin{equation} \label{eq:Zfromdos}
    Z(\beta) = \Tr~\left(e^{-\beta H}\right) = \int dE\,\varrho_H(E)e^{-\beta E},
\end{equation}
and hence to equilibrium thermodynamic observables. At a more microscopic level, 
the distribution of level spacings can distinguish integrable from quantum-chaotic 
dynamics \cite{Oganesyan2007,Pal2010,Atas2013}, while anomalous structures within 
the spectrum may reveal Hilbert-space fragmentation, many-body scars, or the onset 
of localization \cite{Turner2017}. These questions are 
particularly interesting in lattice gauge theories, where local constraints divide 
the Hilbert space into superselection sectors and where non-trivial degeneracies can 
originate from gauge invariance, global symmetries and lattice geometry. A method 
which retains the multiplicities of energy levels and can be restricted to a chosen 
physical sector is therefore of direct interest.

 Obtaining this information classically remains difficult. Exact diagonalization 
(ED) gives the complete spectrum, but consumes memory and resources which scale exponentially 
with the system size. Krylov-space, and particularly tensor-network methods can 
substantially extend the range of accessible systems when only the ground state 
or a few extremal eigenstates are required, but they are not designed to reconstruct 
a dense many-body spectrum. Polynomial-expansion methods, most notably the kernel 
polynomial method, replace ED by the computation of moments and can efficiently 
produce smooth spectral densities for large sparse Hamiltonians \cite{Weisse2006}. 
Stochastic trace estimators further reduce the cost by averaging over random vectors. 
Their accuracy, however, is statistical and depends on the number and quality of 
the sampled states. In lattice field theory, density of states methods such as the 
LLR approach provide a different route to thermodynamics and can overcome severe 
overlap problems by resolving the density over an extensive range of energies or 
actions \cite{Langfeld2016} (see also Sec.~\ref{sec:dos_results}). These methods 
have achieved important successes, but the underlying classical representation of 
the many-body state space remains the limiting factor for Hamiltonian systems with 
a rapidly growing Hilbert space.

Quantum algorithms have approached spectral information from several different directions. 
Quantum phase estimation (QPE) provides a conceptually direct route to an eigenvalue, 
provided one uses a state with sufficient overlap with the corresponding eigenvector 
\cite{Kitaev1995}. Variational and quantum-subspace methods instead construct a small 
effective eigenvalue problem and are naturally adapted to a selected low-energy 
manifold. A second family of algorithms extracts frequencies from real-time data. 
In this approach, a quantum processor measures correlation functions or expectation 
values of the time-evolution operator, while the spectral reconstruction is performed 
classically using Fourier analysis, harmonic inversion or related signal-processing 
techniques \cite{Somma2002,Somma2019}. Methods using probe qubits coupled to a 
Hamiltonian, or performing quench protocols, are closer in spirit to laboratory spectroscopy: 
transitions are identified through the response of the system to a controlled perturbation
\cite{Stenger2022,VilchezEstevez2025}. Finally, quantum implementations of moment-expansion 
methods have been proposed in which Chebyshev moments of the density of states are evaluated 
on quantum hardware \cite{Summer2023}. These approaches target different spectral objects 
and have different requirements on state preparation, coherent evolution and post-processing. 
It is therefore useful to distinguish an algorithm which estimates selected eigenvalues or 
response functions from one which samples the density of states over an extensive portion 
of the Hilbert space.

 In the present work we investigate the latter problem using the maximally mixed-state (MMS) 
protocol introduced in Ref.~\cite{Ferris2023}. The central observation is simple. 
A Hadamard test with the model register prepared in the maximally mixed state (MMS), $\rho_{\rm MMS} = \mathds{1}_D/D$, 
measures the normalised trace of the real-time evolution operator,
\begin{equation} \label{eq:mmsTrace}
\begin{aligned}
    g_{\rm MMS}(t)&= \Tr\!\left(\rho e^{-iHt}\right) = \frac{1}{D}\Tr\!\left(e^{-iHt}\right) \\& = 
    \frac{1}{D}\sum_\alpha d_\alpha e^{-iE_\alpha t},
    \end{aligned}
\end{equation}
where $D$ is the dimension of the sampled Hilbert space. The Fourier transform of this 
signal is the normalised density of states. In contrast to a stochastic average over pure 
states, the maximally mixed state has exactly uniform support over any orthonormal basis 
(not to be confused with a uniform pure state). 
At the end, the signal must still be estimated with a finite number of measurements, 
and the attainable spectral resolution remains controlled by the longest evolution time. 
The output should consequently be understood as a finite-resolution spectral distribution, 
rather than as an efficient classical listing of exponentially many individual eigenvalues.

 Ref.~\cite{Ferris2023} demonstrated this method for a two-qubit Heisenberg Hamiltonian 
on superconducting hardware. Here we extend the protocol to an interacting non-Abelian 
lattice gauge theory (LGT). We consider the $(1+1)$-dimensional $SO(3)$ quantum link model
(QLM) coupled to a single-flavour of staggered fermions, whose gauge-invariant formulation 
provides a compact local encoding while retaining non-trivial matter-gauge interactions 
and a rapidly growing physical Hilbert space. To the 
best of our knowledge, this is the first application of MMS quantum spectral sampling to 
an interacting lattice gauge theory. The gauge model also exposes challenges which are 
largely absent in few-qubit demonstrations: the Hamiltonian contains many non-commuting 
Pauli terms, controlled real-time evolution rapidly becomes deep, and physically useful 
spectra may need to be resolved in definite symmetry sectors.

 Our main result is a comprehensive implementation of quantum spectral sampling adapted to 
a physically-relevant model of strong interactions. We construct controlled Trotterised 
evolution circuits for the $SO(3)$ QLM and develop a Pauli-frame optimization which 
substantially reduces the entangling-gate cost. We benchmark the spectral signal using 
exact diagonalization (ED) and quantum emulators, and demonstrate the protocol on 
superconducting hardware. We further show how 
the same framework can be used to obtain symmetry-resolved spectra 
and coarse-grained densities of states. 
These extensions are important for physical applications: sector resolution 
prevents unrelated symmetry sectors from being superposed in the level statistics, while 
a coarse-grained density of states is the natural object for thermodynamic observables.

 The structure of the paper is as follows. In Sec.~\ref{sec:spec}, we review the MMS-driven 
technique for quantum spectral estimation introduced by \cite{Ferris2023}, as well as 
comparable classical methods. We then show in Sec.~\ref{sec:circuit} how to realise this 
algorithm with an efficient quantum circuit in the case of the $SO(3)$ QLM with a single
flavour of fermions, including a Pauli-frame optimised Trotter step. In Sec.~\ref{sec:nt_results} 
we exhibit the performance of the algorithm at near-term scales, on emulators and on 
superconducting hardware; we also explain how this technique can be naturally 
extended to spectroscopy in a chosen symmetry sector in Sec.~\ref{sec:bn_results}, and 
to the extraction of smooth densities of states in Sec.~\ref{sec:dos_results}. 
Sec.~\ref{sec:conclusion} discusses the consequences of our findings, and explores some 
natural extensions of the work.

\section{Quantum Spectral Sampling}\label{sec:spec}
\subsection{Role of the initial state for spectral information}
 Before discussing particular algorithms, we begin with the aim of estimating the eigenvalues 
$E_\alpha$ of a Hamiltonian $H$. We label the eigenstates $\ket{E_\alpha,a}$, so that
\begin{equation}
    H\ket{E_\alpha,a}=E_\alpha\ket{E_\alpha,a}, \qquad a=1,\ldots,d_\alpha,
\end{equation}
where $d_\alpha$ is the degeneracy of the distinct eigenvalue $E_\alpha$. For an arbitrary 
density matrix $\rho$, the corresponding spectral measure is
\begin{align} \label{eq:SpecMeas}
    \mu_\rho(E) &= \Tr~\left[\rho\,\delta(E-H)\right] \nonumber \\
                &= \sum_{\alpha,a}\bra{E_\alpha,a}\rho\ket{E_\alpha,a} \delta(E-E_\alpha).
\end{align}
 The initial state is thus key to exposing which part of the spectrum is visible 
and with what weight. For a pure state $\rho=\ket{\psi}\bra{\psi}$, Eq.~\eqref{eq:SpecMeas} 
is the local density of states associated with $\ket{\psi}$. A ground-state correlation 
function probes a still more selective object, in which transition energies are weighted 
by matrix elements of a chosen operator. By contrast, the maximally mixed state 
$\rho_{\rm MMS}=\mathds{1}_D/D$ gives
\begin{equation} \label{eq:mmsNormDOS}
    \mu_{\rm MMS}(E)= \frac{1}{D}\varrho_H(E)= \frac{1}{D}\sum_\alpha d_\alpha\delta(E-E_\alpha),
\end{equation}
 which is the normalized density of states. Thus, algorithms which share the same 
 time-evolution component can nevertheless answer physically different questions 
 because they use different initial states and observables.

 It is also important to distinguish between resolving individual eigenvalues vs estimating 
a coarse-grained spectral distribution. A generic many-body Hamiltonian has a number of levels 
which grows exponentially with the volume, and the mean level spacing in the center of the 
spectrum correspondingly becomes exponentially small. No algorithm can explicitly output 
this exponentially long list with polynomial resources. Whole-spectrum algorithms instead 
often return a histogram, a broadened density of states, a finite set of moments, or selected 
spectral features at a prescribed resolution. The resolution scale is particularly relevant for 
cost estimates.

\subsection{Phase-estimation and quantum-subspace methods}
 Quantum phase estimation (QPE) is the canonical quantum algorithm for determining an 
eigenvalue \cite{Kitaev1995}. Let $U(\tau)=e^{-iH\tau}$, and suppose that the input state is
\begin{equation}
    \ket{\psi} = \sum_{\alpha,a}c_{\alpha a}\ket{E_\alpha,a}.
\end{equation}
The QPE circuit applies controlled powers $U^{2^j}$ to imprint the phases $E_\alpha\tau$ 
on an ancilla register. An inverse quantum Fourier transform then produces, schematically,
\begin{equation} \label{eq:qpe_schematic}
    \sum_{\alpha,a} c_{\alpha a} \ket{E_\alpha,a} \ket{\widetilde{E}_\alpha},
\end{equation}
where $\widetilde{E}_\alpha$ denotes the finite-precision energy estimate. A measurement 
returns one eigenvalue with probability $\sum_a|c_{\alpha a}|^2$. QPE is therefore extremely 
powerful when a good approximation to the desired eigenstate is already available \cite{Poulin2018}. 
It can refine an approximate eigenstate, project onto an exact eigenstate and provide an energy 
estimate with a precision controlled by the longest coherent evolution time.

The same properties make QPE less natural as a near-term whole-spectrum method. The probability 
of observing a level is inherited from the initial-state overlap, and states with small overlap 
must be recovered through repeated preparations. Increasing the energy precision typically requires 
longer controlled evolutions. Robust and iterative variants can reduce ancillary-qubit requirements 
and improve tolerance to particular errors, but they do not remove the basic requirements of 
state support and long-time coherent evolution. If QPE is applied to a maximally mixed state, 
every eigenvector is sampled with equal probability and degeneracies are represented correctly; 
however, one eigenvalue is still obtained per measurement, so reconstructing a dense spectrum 
requires extensive repeated sampling.

Variational eigensolvers and quantum-subspace methods make a different compromise. A set of 
trial states $\{\ket{\phi_i}\}$ is prepared, the Hamiltonian and overlap matrices
\begin{equation}
    H_{ij}=\bra{\phi_i}H\ket{\phi_j}, \qquad S_{ij}=\braket{\phi_i|\phi_j},
\end{equation}
are measured, and the resulting generalised eigenvalue problem is solved classically. 
Real-time Krylov methods generate the basis from states of the form $e^{-iHt_j}\ket{\phi_0}$; 
quantum filter diagonalization and variational quantum phase estimation are examples of this strategy
\cite{Parrish2019,Shen2022}. These algorithms can recover several low-lying states with circuits 
shallower than that required for full QPE. Their accuracy is controlled by the span and conditioning 
of the chosen subspace. They are consequently well-suited to a selected part of the spectrum, 
but do not by themselves provide an unbiased estimator of the full density of states.

\subsection{Time-series, moment and spectroscopy-based methods}
 A second route to spectral information is to regard quantum dynamics as a signal-processing 
problem. Starting with the characteristic function, $g_{\rho}$, of the spectral measure,
\begin{align} \label{eq:TrTimeSeries}
    g_\rho(t) &= \Tr~\left[\rho e^{-iHt}\right] 
      = \sum_{\alpha,a} \bra{E_\alpha,a}\rho\ket{E_\alpha,a} e^{-iE_\alpha t},
\end{align}
in the infinite-time limit, the spectral measure follows from
\begin{equation} \label{eq:InvFTmeas}
    \mu_\rho(E) = \frac{1}{2\pi} \int_{-\infty}^{\infty}dt\,e^{iEt}g_\rho(t).
\end{equation}
Somma \textit{et al.} observed that $g_\rho(t)$ can be measured with a single ancilla 
using a Hadamard test \cite{Somma2002}. For a standard choice of phase convention, 
measurements of the ancilla give
\begin{equation} \label{eq:hadamard_trace}
    \braket{X_a} + i \braket{Y_a} = \Tr~\left[\rho e^{-iHt}\right].
\end{equation}
 The measurements at different times are independent: no quantum coherence needs 
 to be preserved between experiments, and the Fourier transform or other reconstruction 
 is carried out classically. Relative to canonical QPE, this greatly reduces the ancilla 
 register and replaces one deep circuit by a set of independent controlled 
 time-evolution circuits.

 The simplest spectrum reconstruction is a discrete Fourier transform. Taking $t_k=k\Delta t$, 
with $|k|\leq K$, a binned estimator can be defined
\begin{equation} \label{eq:FilterEst}
    \widetilde{p}_j = \sum_{|k|\leq K} F_j(-k)\widetilde{g}_k,
\end{equation}
where $\widetilde{g}_k$ is the measured time series and the coefficients $F_j(k)$ define 
the spectral filter associated with the $j$th energy bin. Uniform truncation corresponds to 
the Dirichlet kernel,
\begin{align}     \label{eq:DirichletK}
    D_K(x) & = \frac{1}{2K+1} \sum_{k=-K}^{K}e^{ikx} \nonumber 
    \\
            &= \frac{1}{2K+1} \frac{\sin[(2K+1)x/2]}{\sin(x/2)}, 
\end{align}
with $x=(E_j-E)\Delta t$. The central peak narrows as the maximum time $T=K\Delta t$ 
is increased, but the finite-time truncation also produces side lobes. If an eigenvalue 
does not lie exactly on the discrete Fourier grid, its weight leaks into neighboring bins. 
This \textit{spectral leakage} can shift peak maxima, obscure nearby levels and distort 
degeneracies when the time window is too short. 

 We now move to the description of the MMS spectroscopy, which is a time-series method, 
in the next section. However, we provide a brief description of other quantum spectral-estimation 
strategies in Appendix~\ref{sec:otherSpecStrageies}.

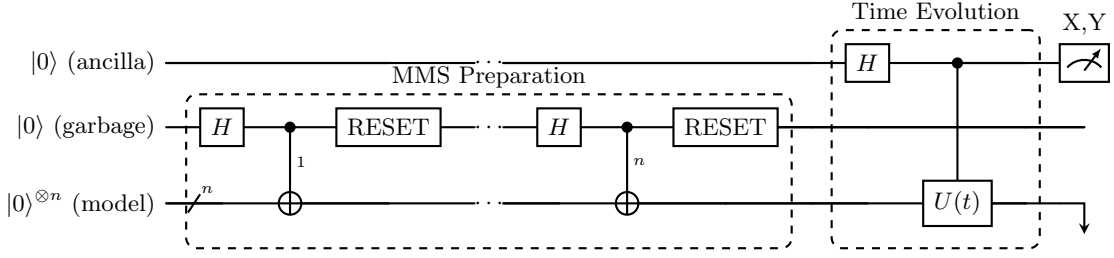
\begin{figure*}[t]   
\centering
\begin{quantikz}[row sep=0.35cm, column sep=0.45cm]
\lstick{$\ket{0}$ (ancilla)} &
\qw & \qw & \qw &
\cdots &
\qw & \qw & \qw & \qw &
\gate{H}\gategroup[
  wires=3,
  steps=3,
  style={dashed, rounded corners, inner xsep=2pt, inner ysep=2pt},
  background
]{Time Evolution} &
\ctrl{2} &
\qw & 
\meter{$X,Y$}
\\
\lstick{$\ket{0}$ (garbage)} &
\gate{H}\gategroup[
  wires=2,
  steps=7,
  style={dashed, rounded corners, inner xsep=2pt, inner ysep=2pt},
  background
]{MMS Preparation} &
\ctrl[wire style={"1"}]{1} &
\gate{\mathrm{RESET}} &
\cdots &
\gate{H} &
\ctrl[wire style={"n"}]{1} &
\gate{\mathrm{RESET}} &
\qw &
\qw & \qw & \qw &
\qw
\\
\lstick{$\ket{0}^{\otimes n}$ (model)} &
\qw\qwbundle{n} &
\targ{} &
\qw &
\cdots &
\qw &
\targ{} &
\qw &
\qw &
\qw &
\gate{U(t)} &
\qw &
\trash{}
\end{quantikz}
\vspace{0.3cm}
\caption{Quantum circuit for maximally mixed-state spectral sampling using a reusable garbage qubit. 
The garbage qubit is entangled sequentially with each model qubit, measured with the outcome ignored, 
and reset before the next step. This prepares the model register in the maximally mixed state while 
reducing the qubit count from $2n+1$ to $n+2$, including the Hadamard-test ancilla.}
\label{fig:spectral_circuit}
\end{figure*}

\subsection{Maximally mixed-state quantum spectroscopy}
 We finally specialise to the protocol used in this work. Let the sampled Hilbert space have
dimension $D$. For an $n$-qubit register without additional constraints, $D=2^n$ and
\begin{equation} \label{eq:MMS_rho}
    \rho_{\rm MMS} = \frac{\mathds{1}_{2^n}}{2^n}.
\end{equation}
Substituting Eq.~\eqref{eq:MMS_rho} into Eq.~\eqref{eq:TrTimeSeries} gives
\begin{align}
    g_{\rm MMS}(t) = \frac{1}{D} \Tr~\left(e^{-iHt}\right) 
                   = \frac{1}{D} \sum_\alpha d_\alpha e^{-iE_\alpha t}.
    \label{eq:mms_trace_signal}
\end{align}
Every eigenvector therefore enters with the same weight $1/D$, and a level with 
degeneracy $d_\alpha$ appears with weight $d_\alpha/D$. The protocol estimates 
the characteristic function of the normalised density of states without requiring 
an eigenstate-preparation routine.

To see how the reduced ancilla state arises, write the controlled evolution as
\begin{equation}
    C_U = \ket{0}\bra{0}_a\otimes \mathds{1} + \ket{1}\bra{1}_a\otimes U, \qquad U=e^{-iHt}.
\end{equation}
Starting from $\rho_{\rm tot}(0)=\ket{+}\bra{+}_a\otimes\rho$, one obtains
\begin{align}
    \rho_{\rm tot}(t) &= C_U \rho_{\rm tot}(0) C_U^\dagger \nonumber\\
                      &= \frac12 \left[\ket0\bra0\otimes\rho + \ket0\bra1\otimes\rho U^\dagger \right. \nonumber\\
                      &+ \left.  \ket1\bra0\otimes U\rho + \ket1\bra1\otimes U\rho U^\dagger \right].
\end{align}
Equivalently, in the ancilla basis we can write
\begin{equation}
    \rho_{\rm tot}(t)= \frac12 \begin{pmatrix}
                                \rho & \rho U^\dagger\\
                               U\rho & U\rho U^\dagger
                               \end{pmatrix},
\end{equation}
where each entry is an operator on the model register. The action of the partial trace over the model register is
\begin{equation}
    \Tr_{\rm model} \left[ \ket i\bra j_a\otimes A \right] = \ket i\bra j_a\,\Tr(A).
\end{equation}
Therefore
\begin{align}
    \rho_a(t) &= \Tr_{\rm model}\rho_{\rm tot}(t) 
              = \frac12  \begin{pmatrix}
                           \Tr(\rho) & \Tr(\rho U^\dagger)\\
                           \Tr(U\rho) & \Tr(U\rho U^\dagger)
                         \end{pmatrix}.
\end{align}
Using $\Tr\rho=1$, $U^\dagger U=\mathds{1}$, and cyclicity of the trace, this becomes
\begin{equation}
    \rho_a(t) = \frac12 \begin{pmatrix}
                          1 & g_\rho^*(t)\\
                          g_\rho(t) & 1
                        \end{pmatrix},
    \qquad g_\rho(t)=\Tr[\rho U(t)].
\end{equation}
For the maximally mixed state,
\begin{equation}
    g_{\rm MMS}(t) = \Tr[\rho_{\rm MMS}U(t)] = \frac1D\Tr[U(t)].
\end{equation}


 To prepare the MMS using a pure state quantum circuit, one may introduce an auxiliary
garbage register and perform the purification
\begin{equation} \label{eq:mmsPurification}
    \ket{\Phi_D} = \frac{1}{\sqrt D} \sum_{x=0}^{D-1} \ket{x}_{\rm model}\ket{x}_{\rm garb},
\end{equation}
and then trace over the garbage qubits to produce the MMS,
\begin{equation}
    \Tr_{\rm garb} \left(\ket{\Phi_D}~\bra{\Phi_D}\right) = \frac{\mathds{1}_D}{D}.
\end{equation}
For $D=2^n$, Eq.~\eqref{eq:mmsPurification} is prepared as a product of $n$ Bell pairs 
between the model and garbage registers. 
The garbage qubits are not disturbed during controlled time evolution and are 
discarded at the end of the computation. If mid-circuit measurement and reset are available, the same 
physical garbage qubit can be reused: it is entangled with one model qubit, measured 
with the outcome ignored, reset, and then used for the next model qubit. Since each 
step leaves the corresponding model qubit in $\mathds{1}_2/2$, the final reduced state 
is again $\mathds{1}_{2^n}/2^n$. This reduces the number of required qubits from $2n+1$ to $n+2$, 
including the Hadamard-test ancilla.

The complete circuit is shown in Fig.~\ref{fig:spectral_circuit}. The pointer ancilla 
is prepared in $\ket{+}$, the model register is prepared in the MMS, and a controlled 
implementation of $U(t)=e^{-iHt}$ is applied. Repeated measurements of $X_a$ and $Y_a$ 
yield the real and imaginary parts of $g_{\rm MMS}(t)$. In practice, we evaluate this 
signal at a finite set of times,
\begin{equation}
    t_k=k\Delta t, \qquad k=0,\ldots,N_t-1,
\end{equation}
and reconstruct the spectrum classically. While $n$ distinct garbage qubits may be used 
for the MMS preparation step, we note that mid-circuit measurements are increasingly 
available on quantum hardware platforms. It is then straightforward to measure/reset 
the same garbage qubit and reuse it for all Bell pairs involved in the MMS; this may 
be shown to be statistically equivalent. We additionally note the nice property that 
because the measurement of the trace only relies on the readout of the ancilla, readout 
error mitigation for real hardware never goes beyond the determination of a $2\times2$ 
confusion matrix, and therefore scales $O(1)$ with system size.

The sampling parameters have direct physical interpretations. The maximum evolution time
$T=(N_t-1)\Delta t$ sets the characteristic energy resolution, $\Delta E\sim \frac{2\pi}{T}$,
while the time step fixes the Nyquist interval,
\begin{equation} \label{eq:nyquist_window}
    |E-E_{\rm shift}| < \frac{\pi}{\Delta t}.
\end{equation}
 An energy shift $H\rightarrow H-E_{\rm shift}\mathds{1}$ may be used to center the 
spectrum inside the available frequency window. Increasing $T$ separates nearby levels but 
requires deeper controlled time-evolution circuits. Reducing $\Delta t$ enlarges the 
unaliased bandwidth but increases the number of time points needed to reach the same $T$. 
These are signal-processing constraints and remain present even for an exact 
implementation of $U(t)$.

 Each estimate of $\Re g_k$ or $\Im g_k$ is obtained from a finite number of binary 
measurements. With $S$ shots per basis and per time point, the statistical uncertainty 
of an individual time-series datum is at most of order $S^{-1/2}$. After a linear Fourier 
reconstruction this produces an approximately frequency-independent background, with an 
amplitude that depends on the transform normalization, the window function and the 
distribution of a fixed measurement budget over the sampled times. The MMS removes the 
additional fluctuations caused by drawing a finite ensemble of random pure states, but 
it does not remove measurement noise.

For comparison, a classical stochastic trace estimator uses random states $\ket{r_k}$ satisfying
\begin{equation}
    \mathbb{E}\left[\ket{r_k}~\bra{r_k}\right] = \frac{\mathds{1}_D}{D}
\end{equation}
and approximates
\begin{equation} \label{eq:ClassicalStrace}
    g_{\rm MMS}(t) \simeq \frac{1}{K} \sum_{k=1}^{K}\bra{r_k}e^{-iHt}\ket{r_k}.
\end{equation}
 The estimator is unbiased, but its variance decreases as $1/K$ and its standard error 
as $1/\sqrt K$, with a prefactor determined by the off-diagonal structure of the operator 
in the sampling basis. For strongly correlated systems or for accurately resolving 
degeneracies, a substantial ensemble may be required. The quantum protocol represents 
the identity density matrix directly and therefore removes this random-vector sampling layer. 
This does not by itself establish an unconditional exponential speed-up: the total cost also 
includes MMS preparation, controlled Hamiltonian simulation, the maximum time required for 
the desired resolution, and the number of measurements needed to recover the signal. It does, 
however, provide a particularly direct quantum representation of the whole-spectrum trace 
and is therefore a natural candidate for spectral sampling in constrained many-body systems.

Finally, the same construction can be restricted to a symmetry sector. If $P_q$ projects 
onto a sector $q$ of dimension $D_q$ and $[P_q,H]=0$, then the sector-resolved signal is
\begin{equation}     \label{eq:SectorResolvedTrace}
    g_q(t) = \frac{1}{D_q} \Tr~\left[P_qe^{-iHt}\right],
\end{equation}
whose Fourier transform gives the normalised density of states in that sector alone. In 
a gauge theory, the relevant projector may encode global charges, lattice symmetries or 
other conserved quantum numbers after the local Gauss-law constraints have already been 
imposed by the encoding. We return to this possibility in Sec.~\ref{sec:bn_results}.

\section{Implementation of the $SO(3)$ Quantum Link Model}\label{sec:circuit}
\subsection{The $SO(3)$ QLM with adjoint quarks}
\label{sec:so3_qlm_spin32}
 We consider the one-dimensional $SO(3)$ QLM introduced in Ref.~\cite{Rico2018}. 
The matter fields are single-component fermionic quarks $\psi_x^a$, with colour index $a=1,2,3$, transforming 
in the adjoint representation of $SO(3)$. The gauge degrees of freedom live on links $(x,x+1)$ 
and are described by quantum link operators $O^{ab}_{x,x+1}$.  In contrast to Wilson's formulation, 
the link Hilbert space is finite-dimensional, while retaining the exact gauge symmetries of the original
theory.  This finiteness is the feature which makes the model particularly natural for quantum simulation and 
for the spectral-sampling protocol considered in this work.

In one spatial dimension the Hamiltonian contains no plaquette term, and can be written as 
\begin{align}
 H_{\rm QLM} =& -t \sum_x \sum_{a,b=1}^{3} \left( \psi_x^{a\dagger} O^{ab}_{x,x+1} \psi_{x+1}^{b}
              + {\rm H.c.} \right) \nonumber\\&+ m \sum_x (-1)^x M_x + G \sum_x \left( M_x-\frac32 \right)^2\nonumber\\
             & 
              + V \sum_x \left( M_x-\frac32 \right) \left( M_{x+1}-\frac32\right),
 \label{eq:so3_qlm_hamiltonian}
\end{align}
where 
\begin{equation} \label{eq:mes}
M_x = \sum_{a=1}^{3} \psi_x^{a\dagger}\psi_x^a
\end{equation}
counts the number of staggered fermions on site $x$.  The hopping term moves a fermion 
between neighboring sites while parallel-transporting its colour through the quantum link.  
The staggered mass $m$ explicitly breaks the one-site shift symmetry of the staggered fermion 
formulation, while the $G$ and $V$ terms are gauge-invariant on-site and nearest-neighbour density
interactions.  The conserved baryon number is
\begin{equation}
B = \sum_x \left( M_x-\frac32 \right).
\label{eq:so3_baryon_number}
\end{equation}

 The fermions follow the usual anticommutation relations: $\left\{ \psi^a_i, \psi^b_j\right\}=
\left\{ \psi^{a \dagger}_i, \psi^{b \dagger}_j\right\} = 0$, and 
$\left\{ \psi^a_i, \psi^{b \dagger}_j\right\}= \delta_{ab} \delta_{ij}$. In addition
to the gauge links, $O^{ab}_{x,x+1}$, the other dynamical degrees of freedom are the 
left and the right electric fluxes, $L^a_{x,x+1},R^a_{x,x+1}$, which satisfy appropriate 
commutation relations with the link operators,
\begin{align}
   [L^a, L^b] &= 2 i \varepsilon^{abc} L^c, [R^a, R^b] = 2 i \varepsilon^{abc} R^c, [L^a, R^b]=0\nonumber\\
[L^a, O^{bd}] &= 2 i \varepsilon^{abc} O^{cd}, [R^a, O^{bd}] = - 2 i O^{dc} \varepsilon^{cba}, \nonumber
\end{align}
sufficient to ensure that the exact gauge symmetry is realized irrespective of the local Hilbert 
space dimension of the operators $L,R,O$. For the minimal $SO(3)$ quantum link representation, 
the link operator may be represented by two spin-$\frac{1}{2}$ objects, one associated with each 
end of the link,
\begin{align}
O^{ab}_{x,x+1} &= \sigma^a_{x,+} \otimes \sigma^b_{x+1,-}, \nonumber \\
L^a_{x,x+1} &= \sigma^a_{x,+}\otimes\mathbb{I},~R^a_{x,x+1} = \mathbb{I}\otimes\sigma^a_{x+1,-}.
\label{eq:so3QLrep}
\end{align}
 The subscripts $+,-$ denote the location of the links with respect to the sites: 
a $+$ is the forward end, while a $-$ is the backward end, and is unique in one spatial 
dimension. The constrained 
nature of the problem lies in the selection of the states which are classified as 
\emph{physical}. In a gauge theory, this is typically done via the imposition of the
Gauss' Law, which generate local gauge transformation and annihilates physical states.
For this theory, the gauge transformations are generated locally by
\begin{equation}
G_x^a = \psi_x^{b\dagger} T^a_{bc} \psi_x^c + L^a_{x,x+1} + R^a_{x-1,x},
\label{eq:so3_gauss_generator}
\end{equation}
where  $T^a_{bc}=-2i\epsilon_{abc}$ are the $SO(3)$ generators in the adjoint representation.
Physical states obey Gauss' Law, 
\begin{equation}
G_x^a\ket{\Psi}=0, \qquad a=1,2,3, \qquad \forall x. 
\label{eq:so3_gauss_law}
\end{equation}
Thus, the gauge-invariant Hilbert space consists only of local $SO(3)$ singlets formed from 
the matter field at a site and the adjacent quantum-link degrees of freedom.

 A common problem with quantum simulation of gauge theories is the overuse of quantum resources
to represent the different degrees of freedom, although only a tiny subset is physical. 
In some cases it is possible to impose the Gauss' Laws analytically to devise a
gauge-invariant encoding, targeting only the physical states, reducing the circuit depth, 
and consequently the errors from quantum hardware. 

The key attractive feature of this model is the very natural gauge-invariant encoding 
available, first developed in \cite{Rico2018}, and extended to $(2+1)d$ in \cite{VanGoffrier2024}. We briefly review the construction here,
since we use the gauge-invariant formulation to estimate the spectral density. The most
straightforward gauge-invariant operator is the local fermion occupation number operator
in Eq.~\ref{eq:mes}.
Further, at each site, the matter triplet and the two neighboring link-end degrees of 
freedom must combine to an $SO(3)$ singlet. This gives rise to the further gauge-invariant
operators at each site,
\begin{equation} \label{eq:bar}
    B_{x,+} = \sum_{a=1}^3 \psi^{a \dagger}_x \sigma^a_{x,+},
  ~~B_{x,-} = \sum_{a=1}^3 \psi^a_x \sigma^a_{x,-} .
\end{equation}
The operators $B_{x,\pm}^\dagger$ and $B_{x,\pm}$ can be understood as \emph{baryon}
creation and annihiliation operators. Since it is possible to have fermionic states 
with $M_x=0,1,2,3$, there are locally four gauge-invariant states per site and these
creation and annihilation operators raise the baryon number of the sites. The matrix
representation of these operators are straightforward,
\begin{equation}
M_x = \begin{pmatrix}
       0 & 0 & 0 & 0\\
       0 & 1 & 0 & 0\\
       0 & 0 & 2 & 0\\
       0 & 0 & 0 & 3
      \end{pmatrix},
      ~
B_{x,+} = \begin{pmatrix}
           0 & \sqrt{3} & 0 & 0\\
           0 & 0 & 2 & 0\\
           0 & 0 & 0 & \sqrt{3}\\
           0 & 0 & 0 & 0
          \end{pmatrix},
\label{eq:GIops}
\end{equation}
with $B_{x,-}= -P_x B_{x,+}$, where $P_x = (-1)^{M_x} = \textrm{diag}(1,-1,1,-1)$ acts as a
parity operator sensitive to the number of fermions in a state. 

Once the gauge-invariant basis is constructed, the local four-dimensional local Hilbert space 
can be identified with that of a single spin $S=3/2$.  Choosing the ordered basis 
$M_x=0,1,2,3 \longleftrightarrow S_x^3=-\frac32,-\frac12,\frac12,\frac32$, we can identify
$M_x = S_x^3+\frac32$. The matrix elements in Eq.~\eqref{eq:GIops} are precisely the Clebsch-Gordan 
factors of the spin-$3/2$ ladder operators, up to phase factors. Thus, the $B_{x,-}$ and 
$B^\dagger_{x,-}$ operators can be exactly mapped to the spin-$3/2$ raising and lowering 
operators, with additional sign factors. 
Substituting the encoded operators into Eq.~\eqref{eq:so3_qlm_hamiltonian} gives
\begin{align}
H_{\rm sp} &= -t \sum_x \left( B^\dagger_{x,+} B_{x+1,-} + B_{x,+} B^\dagger_{x+1,-}\right) \nonumber\\
           &+ m \sum_x (-1)^x M_x + G \sum_x \left(M_x - \frac{3}{2}\right)^2 \nonumber\\ 
           &+ V \sum_x \left(M_x - \frac{3}{2}\right) \left(M_{x+1} - \frac{3}{2}\right).
\label{eq:spin32_hamiltonian}
\end{align}
%
Thus, after imposing Gauss' Law exactly, the non-Abelian $SO(3)$ lattice gauge theory 
in one spatial dimension becomes a local spin-$3/2$ chain of dimension $L$. 
The $L$ sites correspond spin $1/2$'s numbering $Q=2L$.  The matter-gauge hopping 
maps to an XY exchange, the staggered mass maps to a staggered longitudinal field, 
the on-site density interaction maps to a single-ion anisotropy, and the nearest-neighbour 
density interaction maps to an Ising coupling.  In this encoded description gauge invariance
is not enforced by an energy penalty: gauge-variant states have been removed from the 
Hilbert space from the outset.

 The matter content of the theory has important consequences for the phenomenology. Two kinds of 
\emph{fermionic baryons} are possible in the theory: (a) the usual one formed from \emph{three} 
quarks, as happens in QCD, and (b) a bound state formed from a single (adjoint) quark coupling to
a gluon, which is distinct from QCD. These baryons can then be further bounded in \emph{nuclei}
leading to a rich phase structure in the \emph{toy nuclear physics} in the model, as was 
demonstrated in \cite{Rico2018}. 
In cases where the fermions are chosen in the fundamental representation, the colour-singlet 
excitations give rise to \emph{bosonic baryons}, leading to physics different from QCD. 

\begin{figure}[h!]
    \centering
        \includegraphics[width=0.48\textwidth]{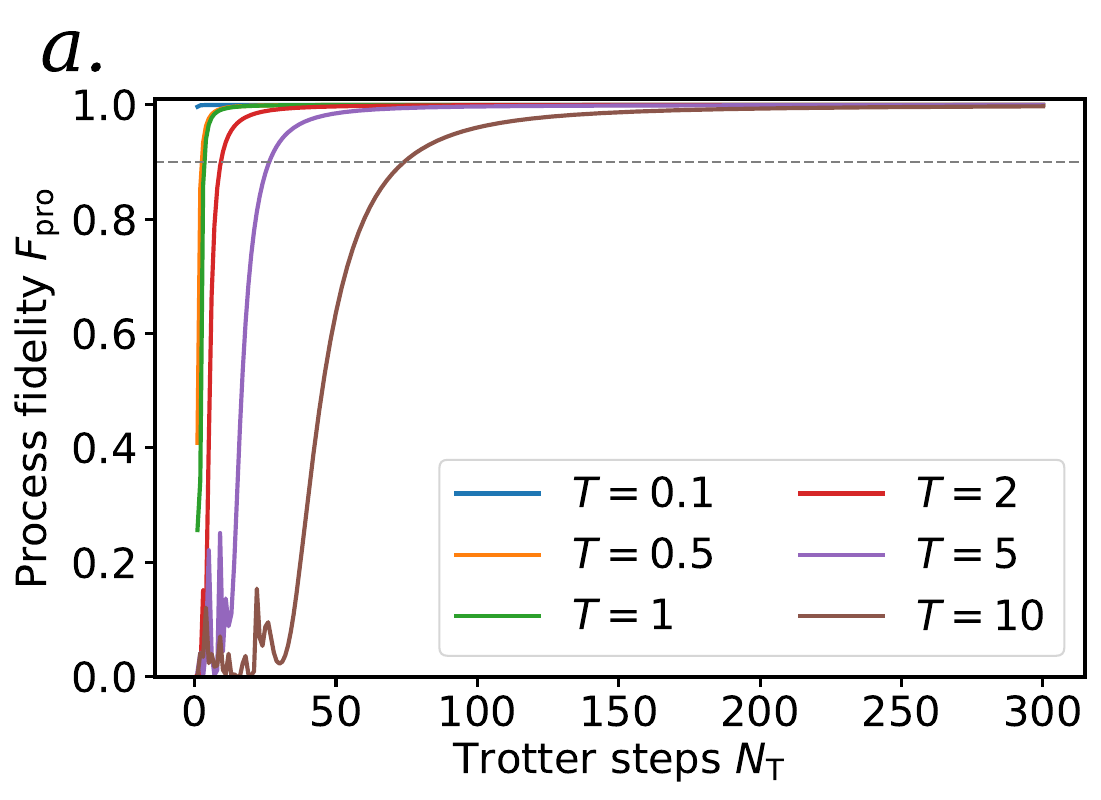}
        \includegraphics[width=0.48\textwidth]{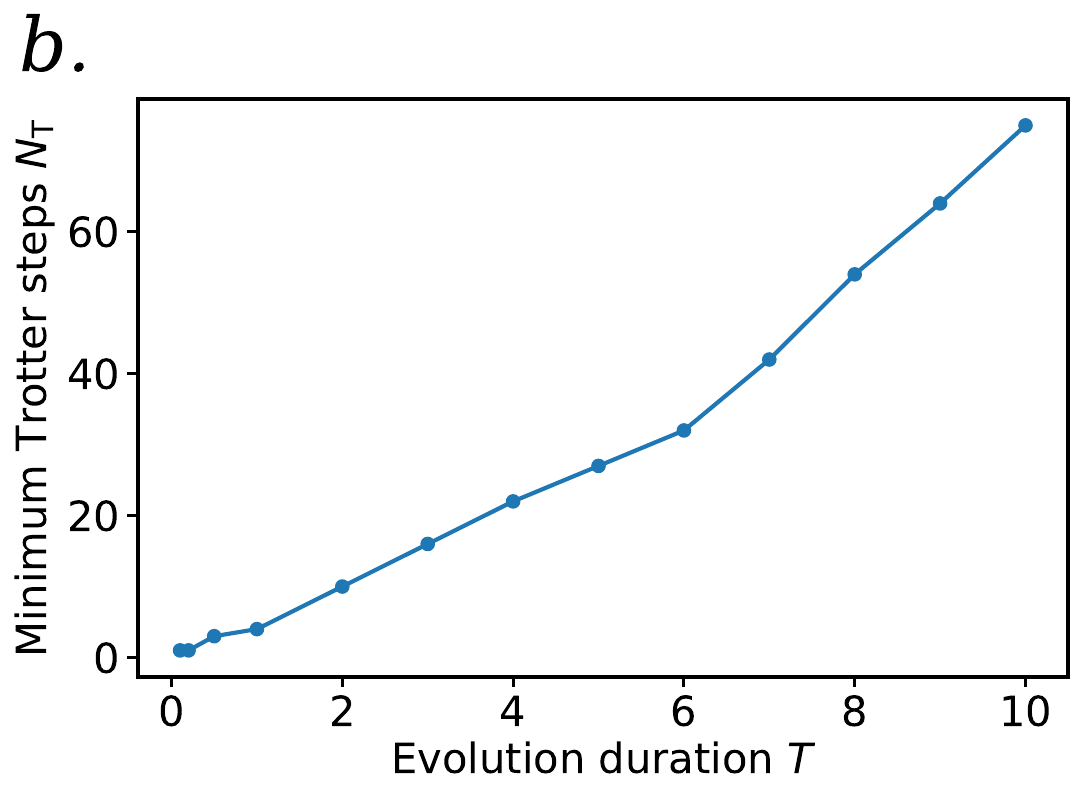}
    \caption{
        Convergence of the first-order Trotter decomposition.
        (\textit{a}) Process fidelity as a function of the number of
        Trotter steps, for several evolution durations $T$.
        (\textit{b}) Minimum number of Trotter steps required to attain
        a process fidelity greater than $0.9$, as a function of $T$.
    }
    \label{fig:trotter1_convergence}
\end{figure}

\subsection{Trotterization of the Hamiltonian}
Trotterization, the decomposition of complicated Hamiltonian time-evolution operators 
in terms of simpler multi-qubit rotations, is based on the observation \cite{trotter1959product} 
that for $H = \sum_i H_i$:
\begin{equation}
    \lim_{m \rightarrow \infty} \left( \prod_i e^{-i H_i t /m} \right)^m = e^{-i H t}.
\end{equation}
Generically, the number of local interaction terms in a lattice Hamiltonian grows with the lattice 
size and with the amount of internal gauge-field information retained by the truncation. For the 
$L=2$ system considered below, corresponding to two spin-$3/2$ lattice sites, the Hamiltonian 
contains $22$ distinct Pauli terms.

In practice, the evolution over a duration $T$ is approximated using a finite number $N_T$ 
of Trotter steps. The accuracy depends both on $N_T$ and on the order $p$ of the product formula:
\begin{equation}
    e^{-iH \cdot T} \simeq \left[ S_p\left(\frac{T}{N_T}\right) \right]^{N_T}.
\end{equation}
Here $S_p(\tau)$ denotes a product formula whose error begins at order $\tau^{p+1}$. For the first-order formula,
\begin{equation}
    S_1(\tau)=\prod_i e^{-iH_i\tau},
\end{equation}
and the leading error of a single step is determined by the commutators between the Hamiltonian terms:
\begin{equation}
    e^{-iH\tau} - \prod_i e^{-iH_i\tau} = -\frac{\tau^2}{2}\sum_{j<k}[H_k,H_j] + \mathcal{O}(\tau^3).
\end{equation}
Consequently, at fixed total duration $T$, the leading global error of the first-order formula 
decreases as $\mathcal{O}(T^2/N_T)$.

To quantify the ideal product-formula error below, we use the process fidelity
\begin{equation}
    F_{\mathrm{pro}}= \frac{ \left| \operatorname{Tr} \left(U^\dagger U_{\mathrm{Trot}}\right) \right|^2}{D^2},
    \qquad D=2^Q,
\end{equation}
where $U=e^{-iH \cdot T}$ and $U_{\mathrm{Trot}}$ is its Trotterised approximation. 
This quantity equals unity when the two unitaries agree up to a global phase. Figure 
\ref{fig:trotter1_convergence} gives a basic illustration of the process fidelity as 
a function of $N_T$ for a variety of time-evolution durations $T$ in Figure
\ref{fig:trotter1_convergence}a, as well as showing a roughly linear increase of 
necessary $N_T$ to obtain a desired process fidelity as a function of $T$ in Figure 
\ref{fig:trotter1_convergence}b. On near-term quantum hardware, however, increasing 
$N_T$ also increases the circuit depth and therefore introduces more decoherence error, 
so a reduction in ideal product-formula error need not produce a more accurate measured result.
\begin{figure}
    \centering
    \includegraphics[width=0.5\textwidth]{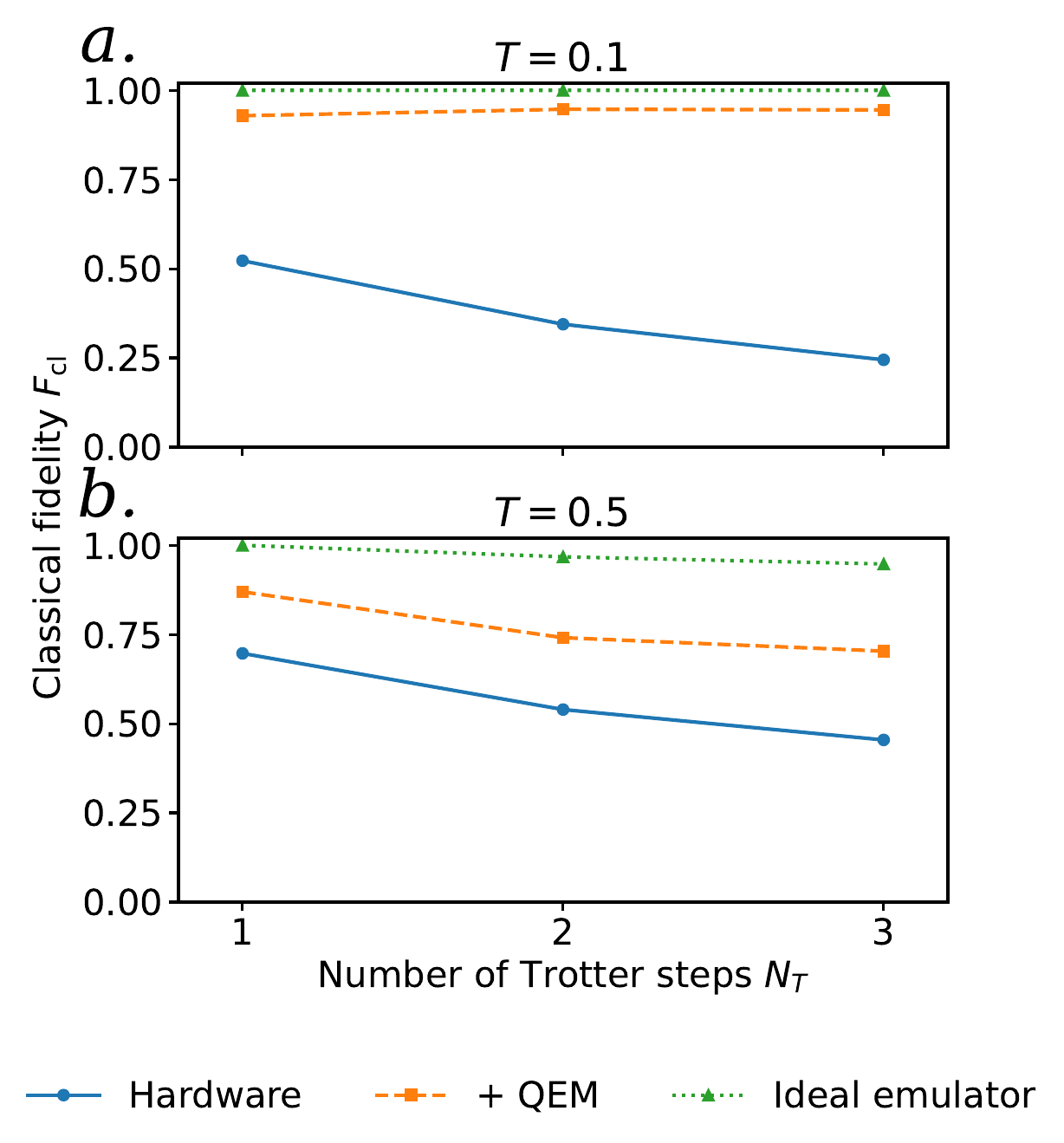}
    \caption{Classical fidelity of the measured output distribution after $N_T=1,2,3$ circuitised 
    Trotter steps on IBM quantum hardware, relative to the exact-theory distribution. Results are 
    shown for (\textit{a}) $T=0.1$ and (\textit{b}) $T=0.5$, comparing the raw hardware data, the 
    hardware data after $T=0$ background subtraction (Hardware + QEM), and an ideal shot-based emulator.}
    \label{fig:ibm_trottersteps}
\end{figure}
We tested these circuitised first-order product formulas on IBM superconducting quantum hardware 
for $T=0.1,0.5$ and $N_T=1,2,3$. Figure~\ref{fig:ibm_trottersteps} shows the classical fidelity 
between the measured and exact output distributions. Although increasing $N_T$ reduces the ideal 
product-formula error, the accompanying increase in circuit depth causes the raw hardware fidelity 
to deteriorate, particularly for $T=0.5$.

We may partially mitigate this deterioration using operator decoherence renormalization (ODR) 
\cite{Farrell2024}. For every $T>0$ circuit, we execute a corresponding calibration circuit with 
the same circuit structure but with $T=0$, whose ideal action is the identity. The deviation of 
this calibration result from its known ideal value is then used to estimate and compensate for 
the noise affecting the corresponding physical circuit. As Figure~\ref{fig:ibm_trottersteps} 
demonstrates, this correction substantially improves the measured distribution fidelity, 
although it does not completely remove the effect of increasing circuit depth.

However, by a combination of further optimizations of our Trotter arrangement and circuitization, 
it is possible to improve matters somewhat. First of all, we can consider the commutators between 
the Hamiltonian interaction terms on a single locality (for us, four neighboring sites). The 
nonzero commutators are sparse enough that is is possible to block these $22$ terms into four blocks, 
where all terms within each block mutually commute; the blocks are diagrammatised in 
Figure~\ref{fig:so3_commuting_graph}. Because Trotter error is driven by the commutators 
between neighboring terms in the expansion, substantial error is avoided by requiring the terms 
in each block to be adjacent, in any order. As we will see below, the permutation of the blocks can 
still substantively impact error, but these permutations are few enough that we can exhaustively 
search them.
\begin{figure}[t!]
    \centering
    \includegraphics[width=0.5\textwidth]{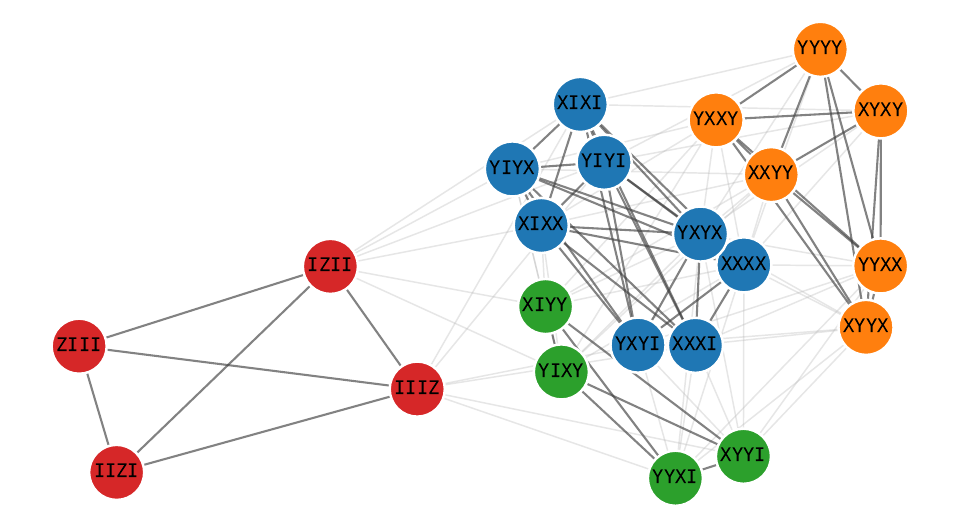}
    \caption{Terms of the $4$-qubit Hamiltonian for the $(1+1)d$ $SO(3)$ QLM, arranged minimally 
    (though not uniquely) into mutually-commuting blocks.}
    \label{fig:so3_commuting_graph}
\end{figure}
Second, we may employ the symmetric higher-order Suzuki product formulas \cite{SUZUKI1990319}. 
These recursively compose lower-order formulas to cancel successive orders in the product-formula 
error; we consider orders $p=2,4,6$, and $8$, in addition to the first-order formula above.

As a clear demonstration of the effectiveness of higher-order Trotter steps in realising longer 
time-evolutions, we show in Figure~\ref{fig:hot_1step_processfidelity} the process fidelity vs. $T$ 
for one Trotter step of order $1,2,4,6,8$. In each case, we consider all $4!$ permutations of the 
aforementioned Hamiltonian blocks; the performance differences between these permutations are notable, 
especially for higher orders, but they are always less significant than the order itself. 

Figure~\ref{fig:hot_1step_processfidelity} first compares the process fidelity obtained from a single 
Trotter step of duration $T$ for product-formula orders $p=1,2,4,6$, and $8$. All results in this 
comparison use the $L=2$ Hamiltonian, corresponding to $L=2$ spin-$3/2$ lattice sites. For each order, 
we test all $4!=24$ permutations of the four commuting Hamiltonian blocks. Increasing the product-formula 
order produces a substantially larger improvement than can be obtained by optimising the block ordering, 
although the latter remains relevant within each order.
\begin{figure}[t!]
    \centering
    \includegraphics[width=0.5\textwidth]{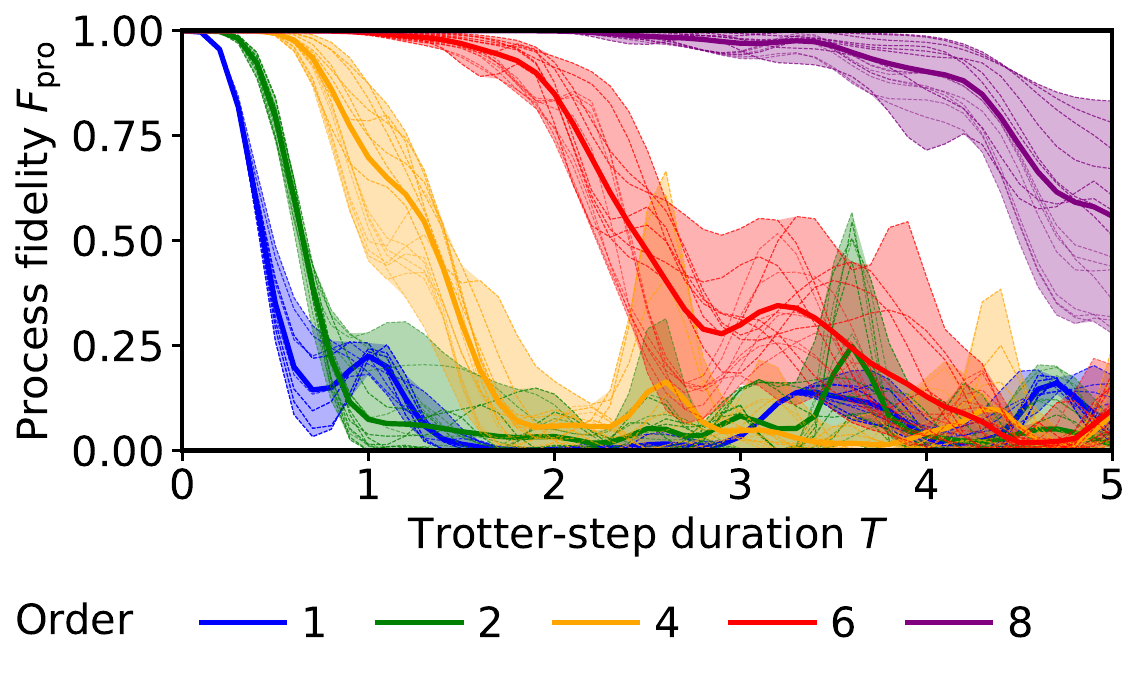}
    \caption{Process fidelity of a single Trotter step as a function of its duration $T$, for 
    product-formula orders $p=1$, $2$, $4$, $6$, and $8$. For each order, the curves show all $4!=24$ 
    permutations of the four commuting Hamiltonian blocks. The separation between orders is greater 
    than the variation induced by the block ordering.}
    \label{fig:hot_1step_processfidelity}
\end{figure}
The reduction in the required number of Trotter steps must, however, be balanced against the increasing 
cost of each higher-order step. After adjacent exponentials of the same commuting block are combined, 
orders $p=1$, $2$, $4$, $6$, and $8$ require respectively
\begin{equation}
    c_p=4,\ 7,\ 33,\ 163,\ 813
\end{equation}
block exponentials per step. We therefore quantify the total cost as

\begin{equation}
    C_{\exp}=c_pN_T,
\end{equation}
where $N_T$ is the minimum number of steps required to reach the target fidelity. 
Figure~\ref{fig:hot_steps_cost_fid0999} shows both $N_T$ and $C_{\exp}$ for the demanding threshold 
$F_{\mathrm{pro}}\geq0.999$. Higher orders require far fewer steps, but their rapidly increasing 
per-step costs produce a sequence of crossovers in the total cost. Low-order formulas are favoured 
for short evolutions, while order $4$ provides the most favorable balance over much of the longer-duration 
range studied. Order $6$ becomes competitive only at still longer durations, and no order $8$ crossover 
is established within the range shown.
\begin{figure}[t]
    \centering
    \includegraphics[width=0.5\textwidth]{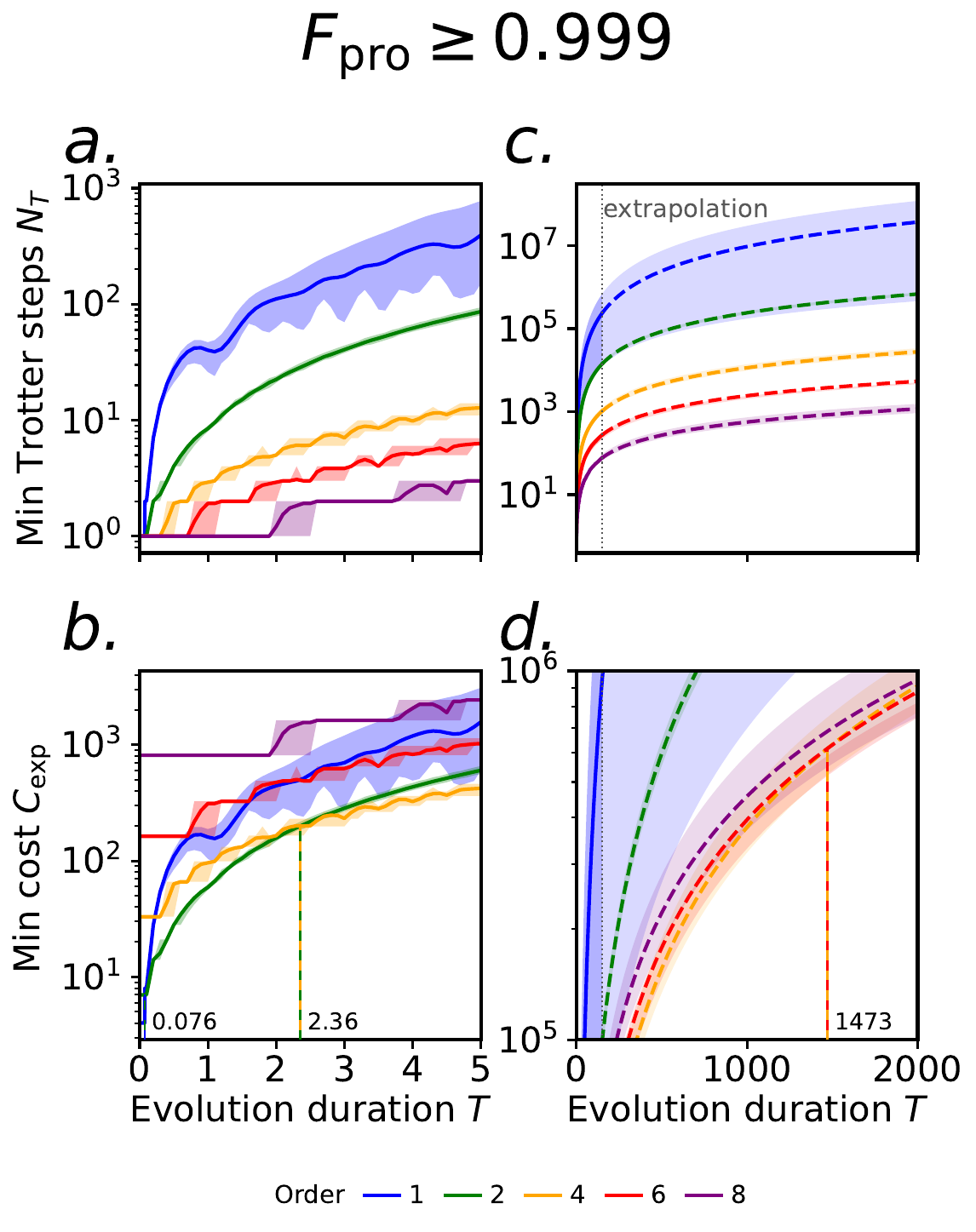}
    \caption{Minimum number of Trotter steps $N_T$ (\textit{a}) and corresponding block-exponential 
    cost $C_{\exp}$ (\textit{b}) required to attain $F_{\mathrm{pro}}\geq0.999$. Curves and shaded 
    bands show the mean and full range over all $24$ commuting-block permutations. Dashed curves in (\textit{c}) and (\textit{d})
    extrapolate the respective power-law fits up to $T=2000$; markers indicate 
    mean-cost crossovers between successive competitive orders.
    }
    \label{fig:hot_steps_cost_fid0999}
\end{figure}
Although detailed spectral-sampling results are reserved for the following sections, 
Figure~\ref{fig:hot_spectra_costmatched} provides a direct test of this cost comparison. We perform 
noiseless state-vector simulations of the maximally mixed state algorithm for $L=2$, using the standard 
higher-order Suzuki product formulas \cite{SUZUKI1990319} and assigning each order approximately 
$1600$ block exponentials. Orders $1$, $2$, and $4$ resolve the principal spectral peaks, whereas the 
high cost of each order $6$ and order $8$ step permits too few steps at this fixed budget to control 
the accumulated product-formula error. For this Hamiltonian, sampling grid, and cost budget, order $2$ 
produces the clearest spectrum, followed closely by order $4$.

We have devised a novel Pauli-Frame-based circuitization approach which is especially well-suited to time evolution with a large quantity of multi-qubit Hamiltonian terms. The background to this approach is given in Appendix~\ref{sec:pauliframe}, and the optimization procedure by which we identified our low-depth Trotter circuit is detailed in Appendix~\ref{sec:pauliframecompilation}. The circuit itself is displayed in Figure~\ref{fig:trotter_circuit}.
\begin{figure}[t!]
    \centering
    \includegraphics[width=0.49\textwidth]{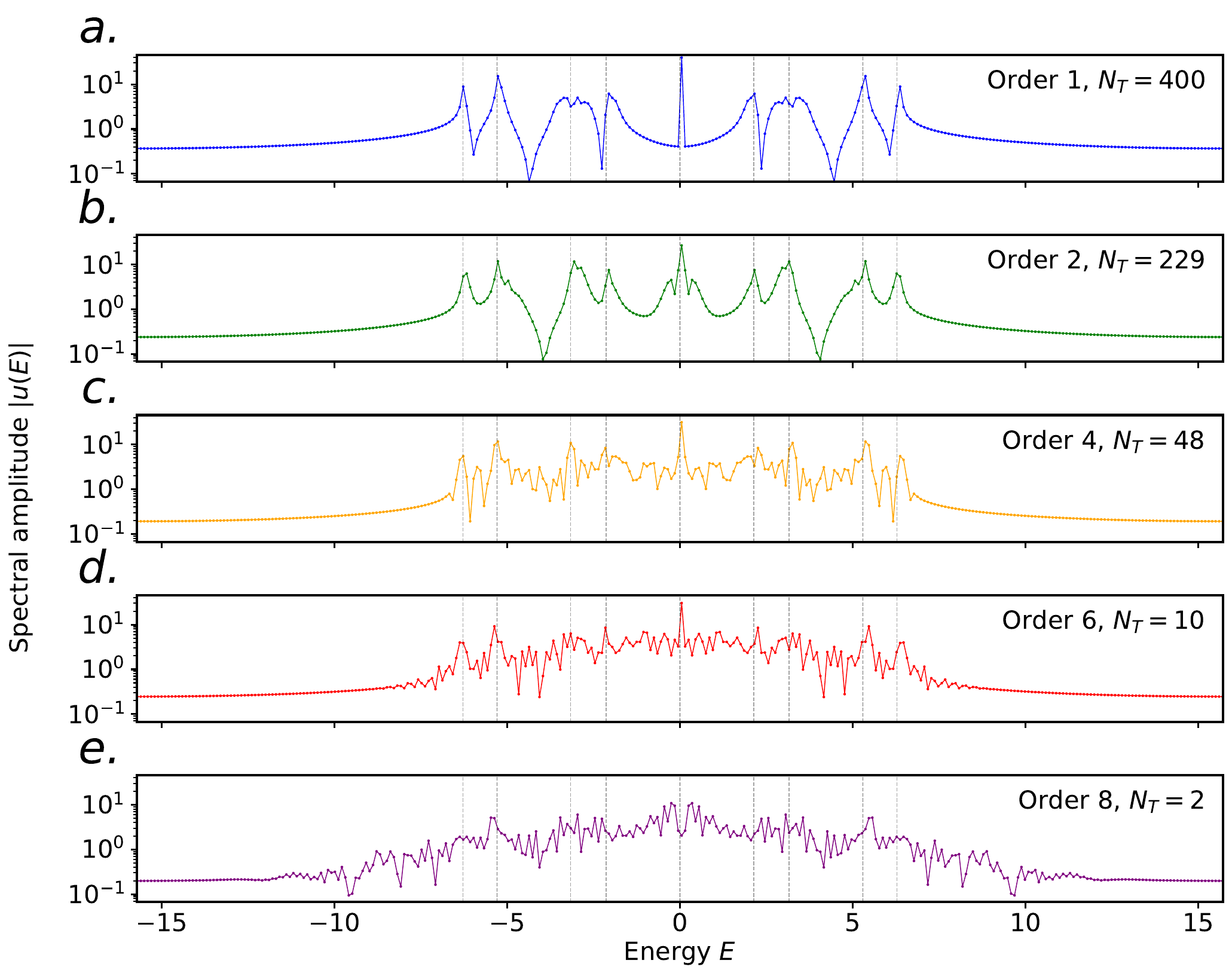}
    \caption{Approximately cost-matched spectral-sampling results for product-formula orders 
    $p=1$, $2$, $4$, $6$, and $8$, using approximately $1600$ block exponentials in each case. 
    From top to bottom, the corresponding numbers of Trotter steps are $N_T=400$, $229$, $48$, 
    $10$, and $2$. All results use the same $L=2$ Hamiltonian and sampling grid. Dashed vertical 
    lines mark the exact eigenvalues of the traceless Hamiltonian.}
    \label{fig:hot_spectra_costmatched}
\end{figure}
%

\section{The Spectral Sampling Algorithm}\label{sec:nt_results}
\subsection{Spectral Sampling Simulator Results}
 We first benchmark the MMS spectral-sampling algorithm against two classical stochastic estimators 
 of the same normalised trace. The controlled time evolution is implemented exactly, without 
 Trotterization, finite-shot sampling, or noise. This isolates the effect of stochastic trace 
 estimation from the other approximations considered elsewhere in this work.

For comparison, we estimate the normalised trace classically as
\begin{equation}
    g_R(t) = \frac{1}{R} \sum_{r=1}^{R}\bra{\phi_r}e^{-iHt}\ket{\phi_r},
\end{equation}
using $R=1000$ random-phase states. Given an orthonormal sampling basis $\{\ket{b_j}\}_{j=1}^{D}$, 
where $D=2^Q$, these states take the form 
\begin{equation}
    \ket{\phi_r} = \frac{1}{\sqrt{D}} \sum_{j=1}^{D} e^{i\theta_{rj}}\ket{b_j},
\end{equation}
where the phases $\theta_{rj}$ are sampled independently and uniformly from $[0,2\pi)$. Figure~\ref{fig:classical_quantum_q4} constructs these states using the eigenbasis of two auxiliary Hamiltonians as the sampling basis: (\textit{a}) the $SO(3)$ QLM in the massless limit, and (\textit{b}) in the quenched ($m \rightarrow \infty$) limit, which is exactly solvable and may be chosen as simply the $Z$-computational basis.
The same exact quantum MMS spectrum is shown in blue in both panels, and all calculations use 
the same time samples and physical energy grid. 

To compare peak amplitudes directly, all spectra are normalised using the quantum peak at $E=0$, accounting for degeneracies. The dotted horizontal lines in Figure~\ref{fig:classical_quantum_q4} consequently 
mark the expected amplitudes $1$, $2$, and $4$. The exact spectrum contains a fourfold-degenerate level 
at $E=0$, twofold-degenerate levels at $E=\pm2$ and $E\simeq\pm6.083$, and nondegenerate levels at 
$E\simeq\pm3.531$ and $E\simeq\pm11.556$.

We identify peaks using an unsupervised procedure applied independently 
to each Fourier spectrum. We search for local maxima in the logarithm of the spectral amplitude, 
heuristically requiring a log-excess of $0.35$ and a minimum separation of two frequency bins. The exact 
eigenvalues are not supplied to the detector, so the dashed exact-diagonalization guides provide 
an independent comparison with the inferred peak locations. Small offsets between the markers 
and these guides are expected from finite Fourier resolution and spectral leakage.

However, the peak amplitudes and therefore degeneracies are more sensitive to the sampling basis. Random-phase sampling in the computational basis reproduces the approximate $1{:}2{:}4$ degeneracy hierarchy but retains visible finite-$R$ fluctuations, including asymmetries between positive- and negative-energy partners. Sampling in the eigenbasis of the massless Hamiltonian agrees substantially more closely with the quantum spectrum and its degeneracies. For this Hamiltonian, at $m=1$, constructing the random-phase states in that basis reduces the finite-$R$ fluctuations relative to the computational-basis construction.

This comparison separates two aspects of the reconstruction. Eigenvalue locations are relatively robust once sufficiently many classical samples are included, whereas relative peak amplitudes, and hence the inference of degeneracies, depend more strongly on the basis used for random-phase sampling. The purified quantum calculation evaluates the MMS trace directly and includes every spectral sector simultaneously. Both classical estimators converge to the same basis-independent trace as $R$ is increased, but their finite-$R$ variances can differ substantially because the off-diagonal matrix elements sampled by the random phases depend on the chosen basis.
\begin{figure}[t!]
    \centering
    \includegraphics[width=0.48\textwidth]
        {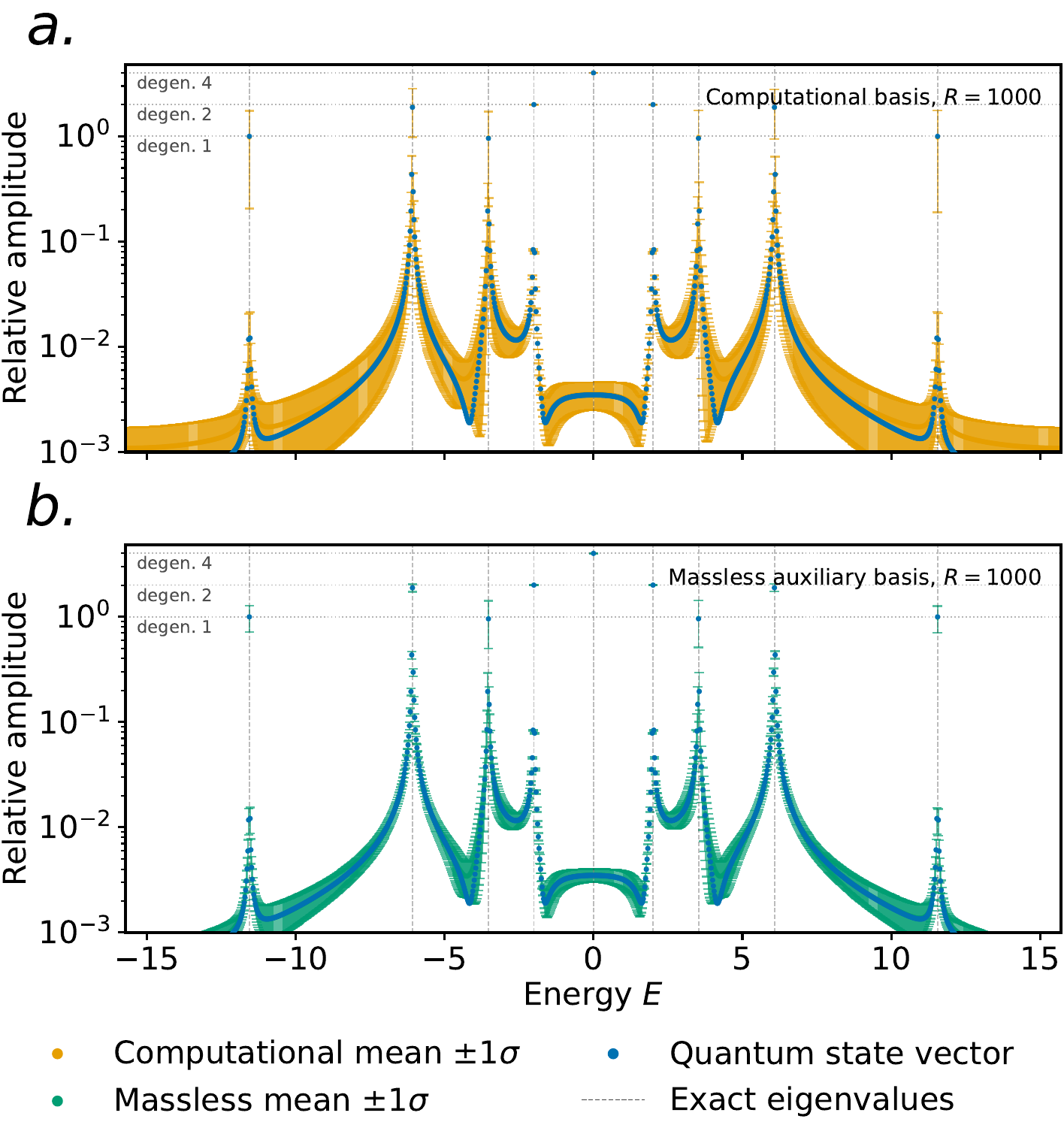}
    \caption{Comparison of random-phase classical trace estimation with exact quantum MMS spectral sampling for the periodic $L=2$ $SO(3)$ QLM, with physics parameters $m=1$ and $t=1$. The classical estimates include means and $1\sigma$-intervals for $R=1000$ random-phase states constructed in the computational basis in (\textit{a}) and in the eigenbasis of the massless $SO(3)$ QLM Hamiltonian in (\textit{b}); dashed vertical lines show the exact eigenvalues and dotted horizontal lines indicate degeneracies $1$, $2$, and $4$.}
    \label{fig:classical_quantum_q4}
\end{figure}

\begin{figure}
    \centering
    \includegraphics[width=0.46\textwidth]{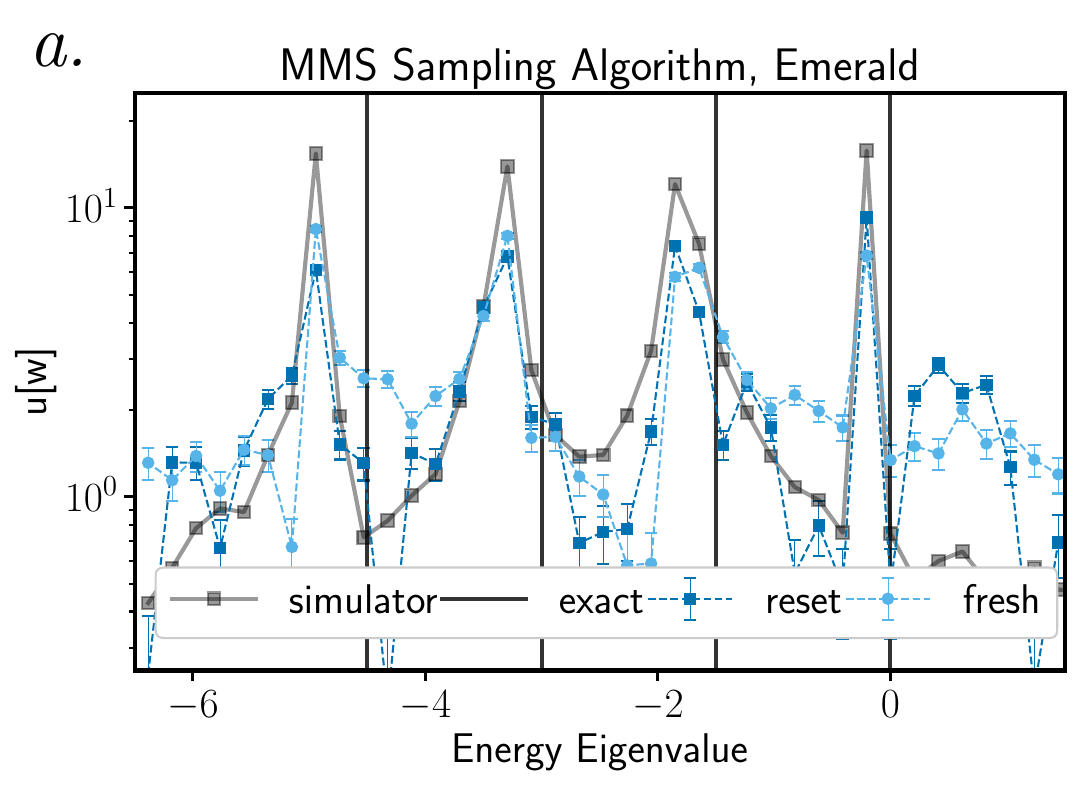}
    \includegraphics[width=0.46\textwidth]{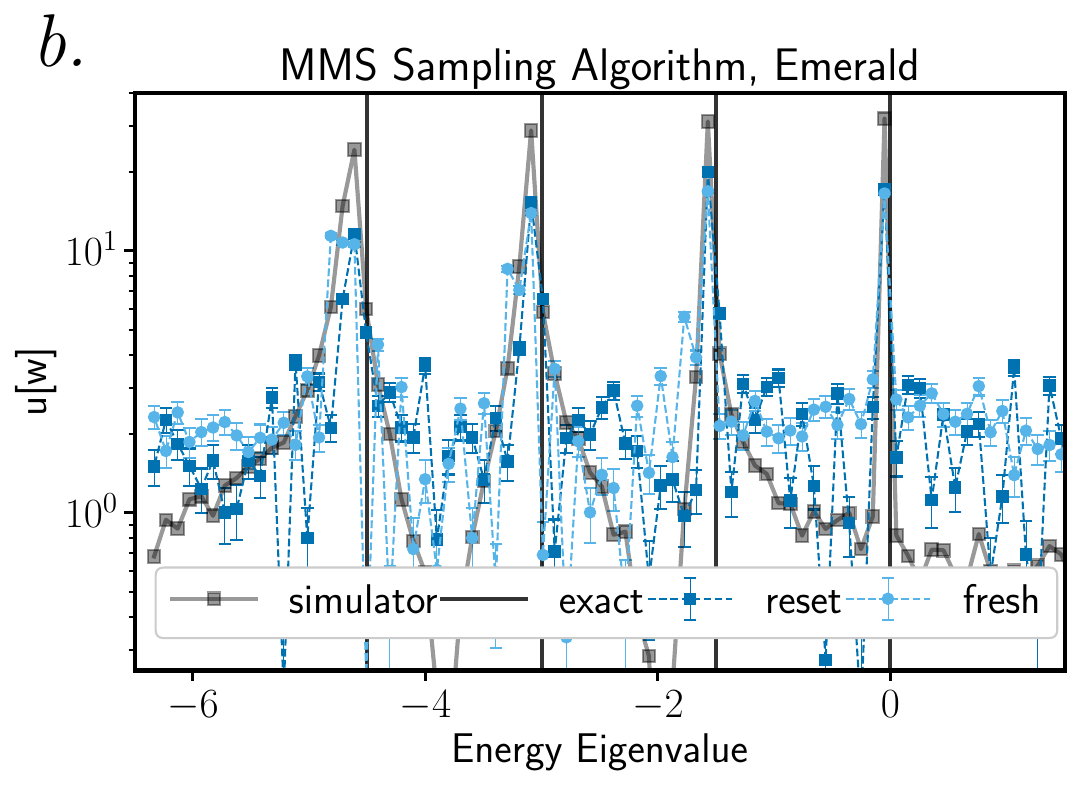}
    \includegraphics[width=0.46\textwidth]{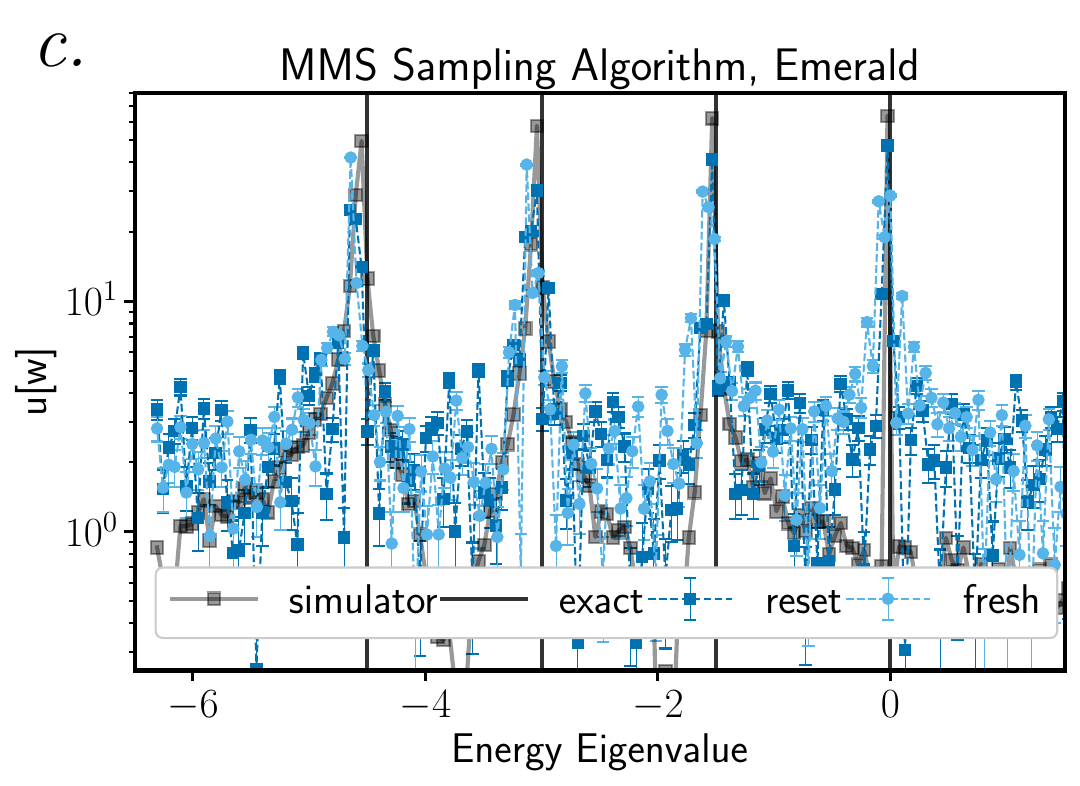}
    \caption{Results of the maximally mixed state algorithm on IQM Emerald. The parameters used are  $\omega_{\mathrm{max}}=2.0$ and (\textit{a}) $d\omega=0.2$, (\textit{b}) $d\omega=0.1$, and 
   (\textit{c}) $d\omega=0.05$. }
    \label{fig:mms_real_examples}
\end{figure}

\subsection{Spectral Sampling Real-Hardware Results}
Here we test the maximally mixed state algorithm for the $L=1, Q=2$ case on IQM's Emerald device, which allows for dynamic circuits--that is, those with mid-circuit measurement. Mid-circuit measurements allow us to reset and reuse the garbage qubit shown in Figure 1, entangling it with each one of the model qubits to form a maximally mixed state. This reduces the number of qubits necessary for a simulation to $Q+2$ from the naive number of $2Q + 1$, effectively doubling the system size that any given quantum computer can study.

Figure~\ref{fig:mms_real_examples} shows the spectrum results after the Fourier transform, which include those for the simulator along with those involving a single garbage qubit that is reset, and those involving the naive number of qubits necessary if we do not have dynamic circuits at our disposal (\textit{fresh}). Finally, they also give peaks at the exact energies for this system. Here we see good agreement between the simulator results and quantum hardware results, regardless of whether fresh or reset qubits are used. The three figures from top to bottom show the effects of increased resolution: from $d\omega=0.2$ (\textit{a}) to $d\omega=0.1$ (\textit{b}) to $d\omega = 0.05$ (\textit{c}). We observe that with better resolution the algorithm shows improved agreement with the exact energies. We also note that $L=1$ requires no Trotterization for time-evolution, as no off-diagonal hopping terms are present, allowing us to calculate evolution at long times without the need for deep circuits. We also note that the systematic left-shifting of the peaks is simply a matter of the algorithm at low resolution rather than any issue of the real hardware, as the peaks for the real hardware line up very will with those from the perfect simulator at each resolution, and the peaks shift closer to the correct energies at each increase in resolution.

\section{Physical Applications of The Algorithm}


\subsection{Disorder and the transition to MBL}
\label{sec:bn_results}
One application of this algorithm involves the ability to extract energy gaps 
from the middle of the spectrum in polynomial-time, rather than the previous 
focus of the low energy spectrum. This allows for studying the transition to many-body localization for a highly disordered model. Here we take advantage of limiting the maximally-mixed state to that of a symmetry-resolved sector of interest: a particular baryon number.

We explore this capability here for the $SO(3)$ model in a couple of baryon number sectors. The baryon number is given by 
\begin{equation}
B = \sum_{j=0}^{L-1} \left(n_j - \frac{3}{2}\right).
\end{equation}
Figure~\ref{fig:disorder} shows data for two system sizes, $L=6$ and $L=4$, diagonalized in the $B=0$ sector, with $t=1$, $m_0=2$, $G=0.3$, and $V=0.13$. The $\alpha$ value is the strength of the disordered coupling, which we choose to be $m$ for these simulations, and the observable is the mean gap ratio $\braket{r}$, averaged over the energy gaps of the middle $50\%$ of energies in the spectrum for $50$  random realizations of the disordered Hamiltonian for $L=6$ and $50$ random realizations for $L=4$.  This mean ratio gives information about which distribution the energy levels follow: 
\begin{equation}
\braket{r} \approx \left\{ \begin{array}{cc}0.5307, & \mathrm{Gaussian\,Orthogonal\,Ensemble\;(GOE)}\\ 0.3863, &\mathrm{Poisson}\end{array}\right. .\nonumber
\end{equation}
When disorder is increased ($\alpha$ is large enough), we expect to see a transition from the chaotic GOE distribution, where the energies \textit{repel} each other and maintain their spacing, to the Many-body Localized (MBL) phase, where the system fails to thermalize. 

Figure~\ref{fig:disorder} shows both the transition of $\braket{r}$ with respect to disorder strength $\alpha$ in (\textit{a}) as well as the distributions of energy level spacings for 50 realizations at $L=6$ overlaid with the Poisson and GOE distributions in (\textit{b}). In Figure~\ref{fig:disorder}(\textit{a}), it is clear that there are finite-size system effects, but even for a system of only $L=6$, one can see the crossover and the two expected values for $\braket{r}$ at different $\alpha$ strengths. Then Figure~\ref{fig:disorder}(\textit{b}) clearly shows the movement to Poisson at large $\alpha$. 

\begin{figure}[h!]
    \centering
    \includegraphics[width=0.46\textwidth]{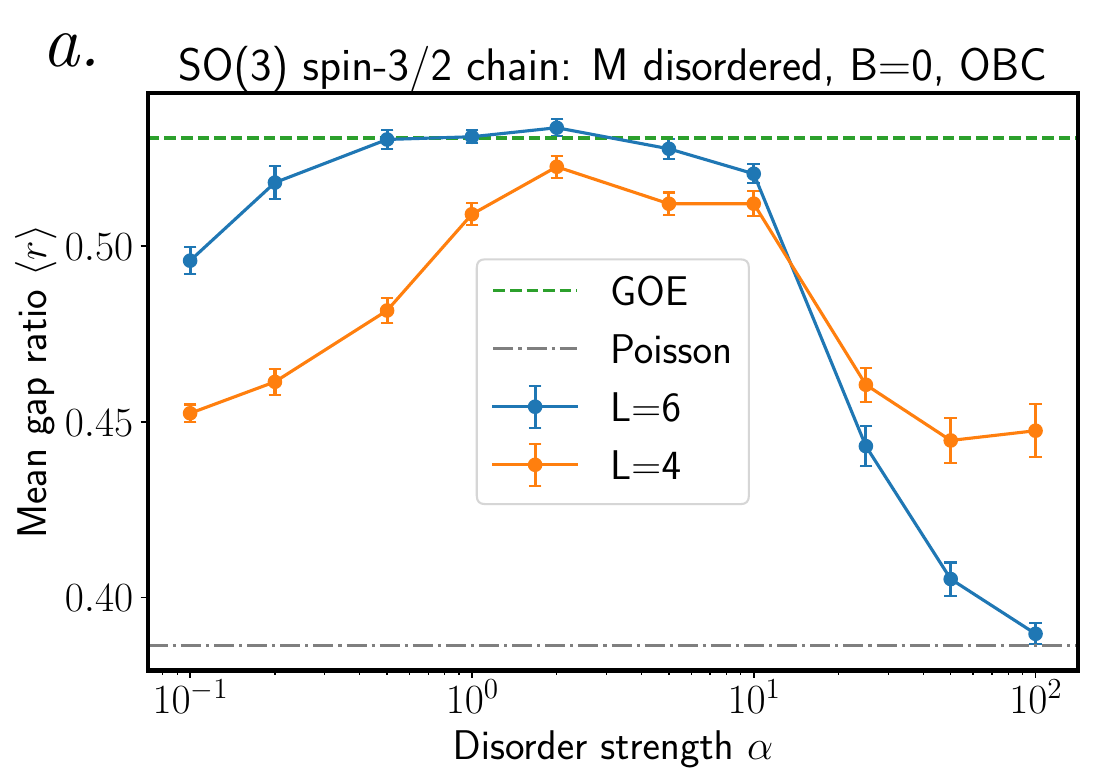}
    \includegraphics[width=0.46\textwidth]{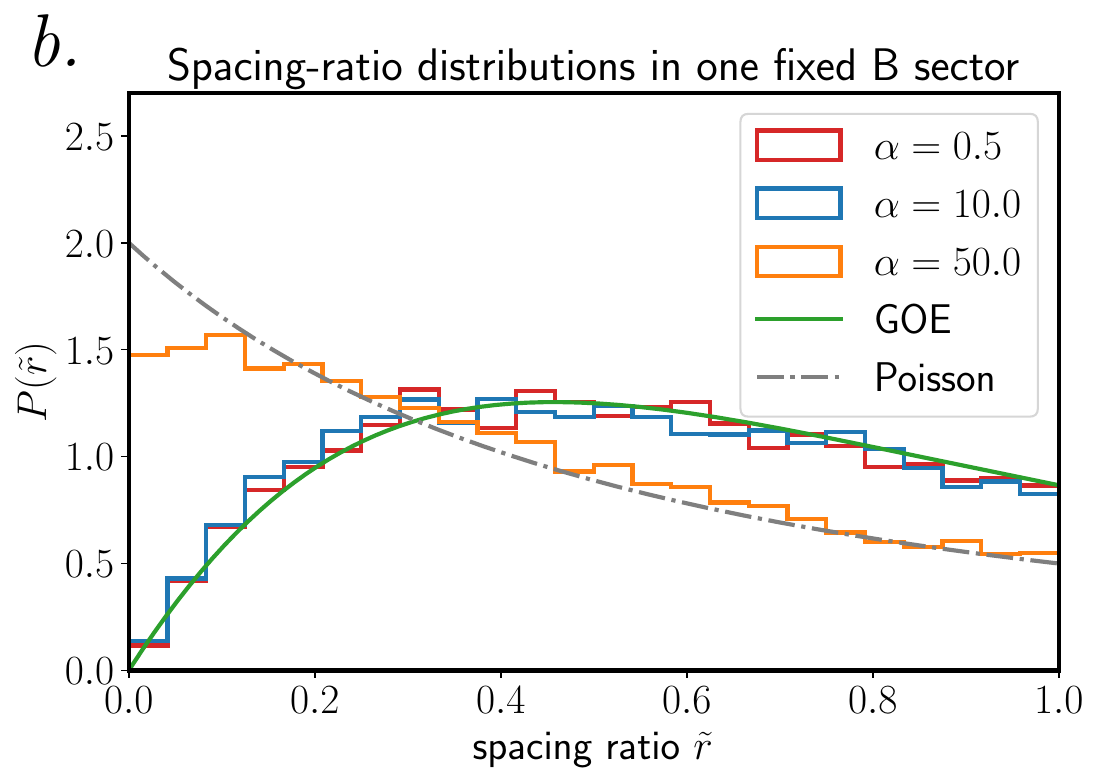}
    \caption{(\textit{a}) The ETH to MBL transition for $L=4$ and $L=6$ as disorder in $m$ is increased. 
    The former lattice shows strong
    deviations from the infinite volume values. (\textit{b}) The level spacing distributions for $L=6$ in the $B=0$ sector.}
    \label{fig:disorder}
\end{figure}

\begin{figure}[h!]
    \centering
    \includegraphics[width=0.49\textwidth]{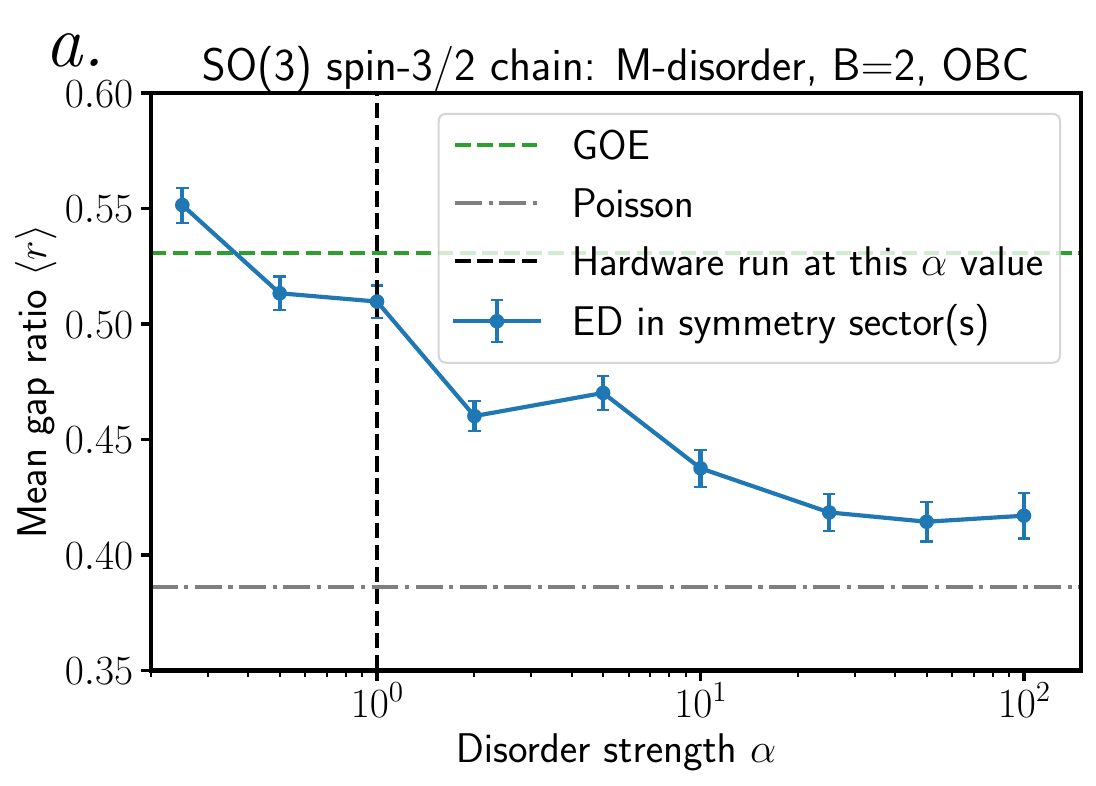}
    \includegraphics[width=0.49\textwidth]{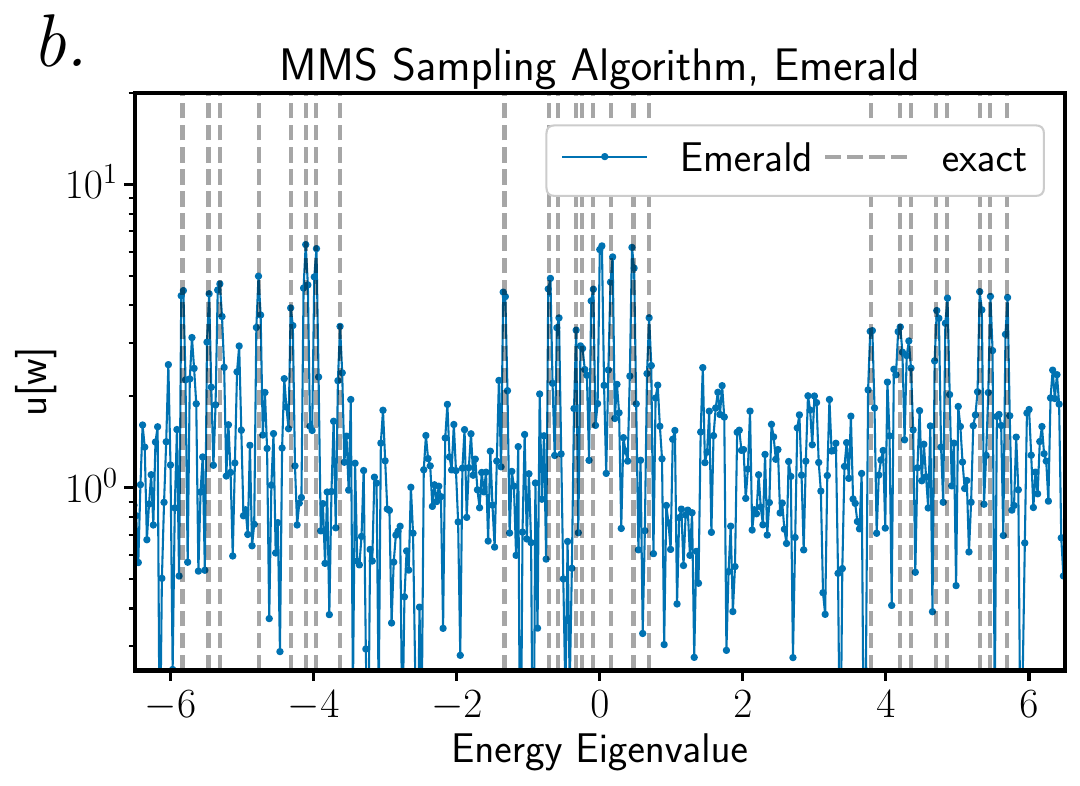}
    \caption{(\textit{a}) The MBL transition for the $B=2$ sector for $L=4$ (simulator). (\textit{b}) Emerald real hardware data for one realization where $\alpha=1.0$.}
    \label{fig:disorder_realhardware}
\end{figure}

For a near-term hardware proof-of-principle of the novel aspects of this algorithm, we resolve the center of the energy spectrum for a larger baryon number sector so as to keep the spectrum small. Thus we also demonstrate the use of a modified maximally mixed state made up of only states in a partular baryon number sector. We emphasize that for sufficiently coherent quantum hardware and a constant number of desired energies to resolve, this algorithm is polynomial-time. To obtain this scaling, one would use Trotterization in order to get necessary amount of time evolution (and thus the middle portion of the spectrum). Here for these near-term calculations however we simply use the diagonal form of the Hamiltonian in order to focus on the spectrum resolution capabilities of near-term hardware alone, since that is the novel aspect of this work. A constant can be added to the spectrum to shift it so that it is centered about zero and we emphasize that this can be estimated in polynomial time \cite{keating2015spectra,mon1975statistical,benet2003review,kota2014embedded,kempe2006complexity,macwilliams1977theory},
 and so the focus for these real hardware calculations will be on the resolution of the energy spacings for the middle of the spectrum.
Figure~\ref{fig:disorder_realhardware} illustrates this small-spectrum calculation where $B=2$ and $L=4$, resulting in a Hilbert space dimension of $31$, fitting nicely within the state space of five qubits. The Hamiltonian used for this calculation is $H_{B=2}\oplus E_{\mathrm{add}}$, where $E_{\mathrm{add}}$ is a number at a different scale from the rest of the spectrum that can be cleanly separated, and is added in order to fill out the state space of the five qubits. The circuits themselves are 78 two-qubit gates deep, and contain 93 two-qubit gates total.  Figure~\ref{fig:disorder_realhardware}(\textit{a}) shows simulator data that computes $\left\langle r\right\rangle$ from the middle $25\%$ of the spectrum for $400$ realizations of disordered mass couplings with $\alpha = 0.25, 0.5, 1, 2, 5, 10,
25, 50, 100$. It shows that it is still possible to see the MBL crossover even with such a small Hilbert space. The bottom panel shows real hardware data from Emerald for the middle of the spectrum for a single realization at $\alpha=1.0$. Dynamical decoupling is used and $T = 209.41$ with $N=712$ steps. There were 430 shots for each measurement. We used a Hann window in the Fourier transform in order to improve the distinction of the peaks. More details are given in Appendix \ref{sec:disorder_runs}. Here it is clear that the hardware is able to resolve the middle of the spectrum, as desired, and would be able to do so in polynomial time once there is enough coherent depth available for sufficient trotter steps. For the spectrum there is a median difference in energy of $0.0043$ between the recovered peaks and the exact peaks; without using the Hann window, the median difference is $0.0142$. Using the middle $25\%$ of the peaks and computing $\left\langle r \right\rangle$ for this random set of disordered couplings gives $ \left\langle r \right\rangle= 0.58\pm 0.11$, compared with $0.54\pm 0.08$ for the exact spectrum. We emphasize that while we should be seeing GOE statistics here, and these numbers appear consistent with GOE statistics, a single set of random couplings is not enough to actually show whether we have GOE or Poisson statistics. This is true for the exact calculation as well as the real hardware calculation. However, with sufficient incidences of real hardware experiments of random couplings, it would be possible to reduce the error and get a conclusive result.

\subsection{Extracting Densities of State}\label{sec:dos_results}

Maximally mixed state spectral sampling also shows promise in augmenting classical state-of-the-art density of states methods with quantum inputs. In Appendix \ref{sec:dos_methods} we give an overview on the background literature on density of states methods for lattice field theories, especially the Wang-Landau and Langfeld-Lucini-Rago (LLR) methods. Here, we discuss how outputs from spectral sampling can enhance the LLR method.

LLR can be effective even in theories with a sign problem, because it replaces the original high-dimensional path integral by a one-dimensional integral over the density of states \cite{Langfeld2012,Langfeld2016}. The nearly constant relative precision of LLR over a large dynamic range can improve the associated overlap problem, since contributions which are exponentially suppressed in an ordinary importance-sampled ensemble may still be determined accurately; this strategy has successfully treated a strong sign problem in the $Z_3$ spin model \cite{Langfeld2016}. LLR does not, however, generically remove the complex sign problem, because the resolved one-dimensional integral may still involve severe oscillatory contributions.

Our proposal is that the Fourier-domain outputs of the quantum spectral sampling algorithm may be postprocessed via an LLR-inspired procedure to estimate a smooth density of states. The usefulness of such an estimate depends on the local level spacing. Taking, for example, a many-body spectrum with a Gaussian coarse envelope,
\begin{equation}
    \varrho_H(E)
    \simeq
    \frac{D}{\sqrt{2\pi}\sigma}
    \exp\!\left[-\frac{(E-\mu)^2}{2\sigma^2}\right],
\end{equation}
the characteristic spacing between neighboring levels is
\begin{equation}
    \Delta_{\rm level}(E)
    \sim
    \frac{1}{\varrho_H(E)}
    =
    \frac{\sqrt{2\pi}\sigma}{D}
    \exp\!\left[\frac{(E-\mu)^2}{2\sigma^2}\right].
\end{equation}
It is smallest near the center, where it scales as $\sigma/D$, and grows rapidly towards the tails. Since $D$ increases exponentially with volume for a generic many-body system, individual levels in the mid-spectrum eventually become impossible to separate at any fixed Fourier resolution. Discrete levels may still be isolated where the resolution permits, but in general the mid-spectrum is naturally represented by a smooth density.

We propose replacing the restricted spectral averages in the LLR midpoint condition by sums over the measured weights in each energy window. If $p_j$ is the measured weight associated with a bin centered at $E_j$, the discrete midpoint condition is
\begin{equation}
    \frac{
        \displaystyle
        \sum_{j\in\mathcal{I}_k}
        (E_j-\bar E_k)\,
        p_j e^{-a_k(E_j-\bar E_k)}
    }{
        \displaystyle
        \sum_{j\in\mathcal{I}_k}
        p_j e^{-a_k(E_j-\bar E_k)}
    }
    =0,
\end{equation}
where $\mathcal{I}_k$ denotes the interval containing $\bar E_k$. The resulting local slopes may then be joined using the same interpolation of $\ln\varrho_H(E)$, with normalization fixed by
\begin{equation}
    \int dE\,\widehat{\varrho}_H(E)=D.
\end{equation}

The window width $\Delta E$ must therefore be large enough to contain meaningful spectral information, and in particular cannot be taken much below
\begin{equation}
    \delta E_{\rm res}\sim\frac{2\pi}{T}.
\end{equation}
For the Gaussian example, the interpolation bias derived above is at most
\begin{equation}
    \frac{\widetilde{\varrho}_H}{\varrho_H}-1
    \simeq
    \frac{\Delta E^2}{8\sigma^2}
\end{equation}
within a centered window; this bias may be reduced at the cost of increasing controlled-evolution time $T$. On the other hand, taking $\Delta E$ too small leaves too little weight in each window and amplifies statistical errors. Errors from Trotterization, hardware noise and spectral leakage also contribute via the measured bin weights.

\begin{figure}[t!]
    \centering
    \includegraphics[width=0.49\textwidth]{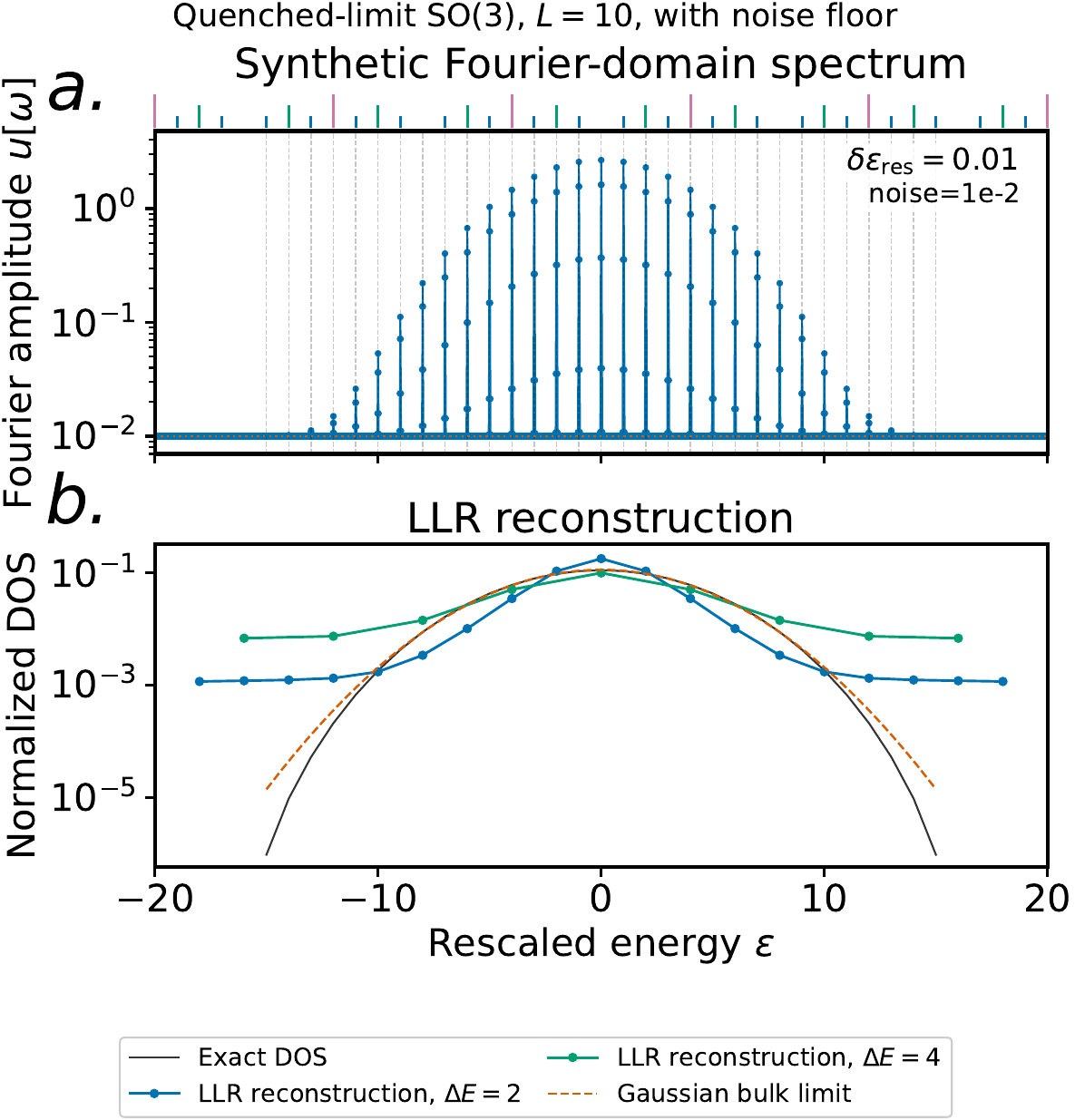}
    \includegraphics[width=0.49\textwidth]{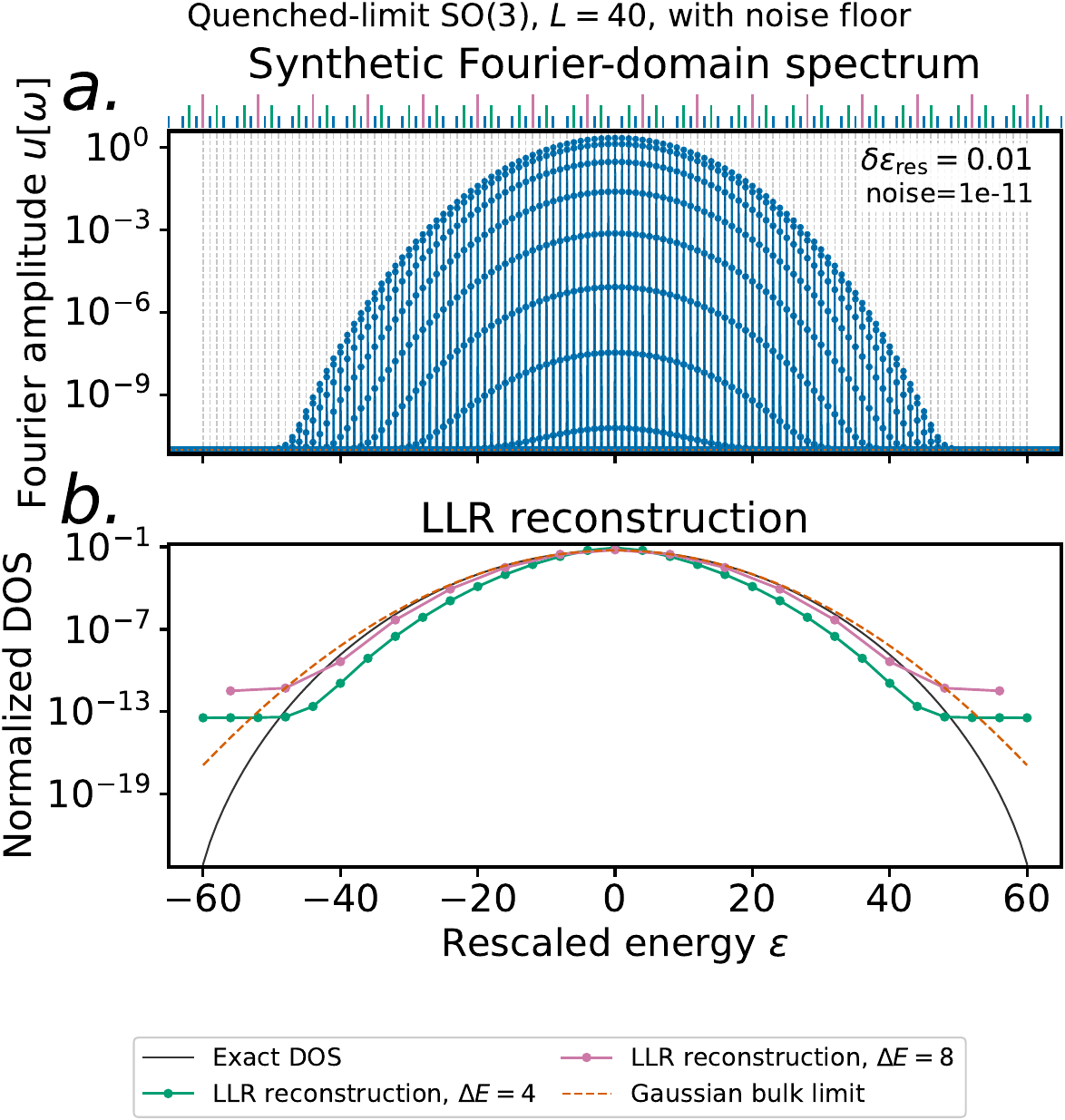}
    \caption{The density of states of the $SO(3)$ QLM in the quenched limit $m \rightarrow \infty$, estimated from synthetic Fourier data via the modified LLR technique, for (\textit{a}) $L=10$ sites and (\textit{b}) $L=40$ sites. In each case, a noise floor is inserted at approximately the geometric midpoint between the minimum and maximum spectral amplitude.}
    \label{fig:synthetic_LLR_10_40}
\end{figure}

For independent measurements and a normalized linear Fourier reconstruction, the uncertainty in a bin weight scales with the total number of measurements as
\begin{equation}
    \sigma_{p_j}
    =
    O\!\left(N_{\rm meas}^{-1/2}\right),
    \qquad
    N_{\rm meas}\sim S N_t,
\end{equation}
when $S$ shots are taken at each of $N_t$ time points, before correlations between neighboring bins are included. If a window contains $n_k$ bins with comparable relative uncertainty $r_k$, weak correlations, and a reweighted variance of order $\Delta E^2$, linearization of the midpoint condition gives the indicative estimate
\begin{equation}
    \sigma_{a_k}
    =
    O\!\left(
        \frac{r_k}{\Delta E\sqrt{n_k}}
    \right).
\end{equation}
These gradient errors accumulate across successive windows, without accounting for covariance effects. The optimum $\Delta E$ is therefore set by a balance between curvature bias, Fourier resolution and statistical noise. Explicating the details of this balance is beyond the scope of this paper, but would be of immense value for applying quantum spectral sampling to larger systems.

We close this section by projecting how our $LLR$ with spectral sampling strategy could perform on the $SO(3)$ QLM when much higher-quality quantum spectral sampling data is available, i.e. via mid-term quantum hardware. In the quenched limit $m \rightarrow \infty$ we have an exact analytical understanding of the model eigenspectrum, which may be chosen to be diagonal in the $Z$-basis. Therefore we are able to generate a clear Fourier spectrum, accounting for degeneracies; this is of course not physically realistic, since it does not account for peak shape, variable broadening, etc. We only include a fixed peak width of $0.01$ in units of energy, which is much smaller than the spectral gaps but larger than the resolution of the synthetic spectrum, and a noise floor set approximately halfway between the minimum and maximum spectral amplitude. Figure 
\ref{fig:synthetic_LLR_10_40} shows the DOS estimations obtained via LLR on this data, for a selection of bin widths; we see that bin width $4$ gives preferable results for $L=10$ sites, while width $8$ gives better accuracy for $L=40$. The numerical results confirm our expectation that the total method delivers its best precision in the mid-spectrum, and rapidly flatlines in the tail regions once the spectral data becomes saturated with noise.

\section{Discussion}\label{sec:conclusion}

We have performed maximally mixed state quantum spectral sampling for an interacting non-Abelian lattice gauge theory: the $SO(3)$-symmetric quantum link model in $(1+1)$d. This model contains nontrivial features such as non-commuting multibody Pauli terms and numerous degeneracies in the energy spectrum. This method makes it possible to obtain the full energy spectrum of the Hamiltonian, or instead, to obtain specific windows of the spectrum using circuits that scale in polynomial time. By initializing with the maximally mixed state, the spectral outputs are unbiased even without any prior physics knowledge of the Hamiltonian.

We explored several strategies for mitigating the cost-scaling of quantum time evolution with the duration required, in particular both higher-order Trotter steps and Pauli-frame-optimised circuit design. Numerical simulations established the time-durations for which each Trotter order gives optimal circuit depth, in particular order $n=4$ for a wider range of physically-relevant time durations. Our circuit compilation relied upon a recursive search over Pauli-frame transitions, which proved to identify much shallower Trotter circuits than comparable heuristics, especially for Hamiltonians with many multi-qubit interaction terms.

We explored the spectral sampling algorithm on real hardware for two cases: the full spectrum of a Hamiltonian for a small system, which can be used for exploring mass gaps and phase transitions, and the mid-energy spectrum of larger systems, which is of use for determining the dynamics of disordered systems. Two features of this algorithm are particularly nice for real hardware calculations. The first feature is that only a single qubit need be measured, allowing for very cheap readout error mitigation (the confusion matrix is only $2\times 2$). The second feature is that only a single additional qubit is necessary to produce the maximally mixed state, because it can be entangled with each of the time evolution qubits in sequence, being reset after each entanglement. This reduces the number of qubits necessary for this work to $2+2L$ rather than $1+4L$, where $L$ is the number of sites in the system. This effectively doubles the accessible system size for a given quantum hardware, or allows for greater optimization in finding the best qubits within a given hardware calibration. 

For the low-energy spectrum case we worked with $L=1$, a two-qubit system where Trotterization is unnecessary, and tested the dynamical circuit capabilities of the IQM Emerald device. We found good resolution of the four energies of this small system and consistent results between the version that reset the garbage qubit and the version that did not (\textit{fresh}), thus showing the viability of the dynamical circuit version on noisy hardware. For the mid-energy spectrum test on real hardware to test the disorder application, we used a baryon sector resolved version of the maximally mixed state, representing a variation on this method that is quite useful for symmetry-limited questions. We performed an energy shift to focus on the middle of the spectrum. This modification allows for the determination of Poisson versus GOE distribution of the energy spacings (and thus the ETH to MBL transition) in a way that scales polynomially on quantum hardware. With five qubits for the time-evolution, one qubit for the control operation, and five more garbage qubits, we constructed circuits on IQM Emerald that were 78 two-qubit gates deep for a simulation in the GOE region. The results for $\left\langle r\right\rangle$ on real hardware agreed with the exact results within errors. 

The last application we introduced is that of extracting densities of states from spectral sampling data. We proposed a hybrid algorithm for treating this output data with the classical LLR machinery developed for lattice field theory. There is a tradeoff to this technique: the window width must exceed the Fourier resolution $2\pi/T$, and is bounded from above by the curvature bias $\Delta E^2/8\sigma^2$, and from below by the statistical error in the reconstructed slopes, which accumulates across windows. Determining the optimum, and the correlations induced between windows by the shared time series, is beyond the scope of this work but is a natural next step.

There are several clear extensions to this work. The maximally mixed spectral sampling method carries over to many types of models, but is well-suited to models that resolve into many symmetry sectors. Many QLMs are thus candidates for the application of spectral sampling, e.g. $SU(2)$ and $SU(3)$ link models for QCD and $U(1)$ for QED. Quantum spectral sampling may be applied naturally in $(2+1)d$, which again is an important physical model scenario. These models may introduce plaquette terms and a correspondingly much richer commuting-block structure, which is exactly the setting where the automated Pauli-frame compiler should provide great savings in circuit depth. Finally, it should be possible to push real hardware to contain more Trotter steps and demonstrate more fully the polynomial-time aspect of the algorithm, since circuits as deep as $\sim 80$ two-qubit gates can now be tolerated by superconducting qubits, as long as simple error mitigation techniques with low shot-cost are implemented. All these directions lead to expanded access to the full \textit{fingerprint} of nontrivial quantum theories.

\section*{Acknowledgments}
 We acknowledge access to IQM Emerald through the AWS Cloud Credit for
Research program and thank Sebastian Stern for developing this partnership 
and for guidance with experiments. We also thank Sebastian Hassinger and David Maurer 
who have also been key in this AWS partnership. We acknowledge access to IBM Eagle through the National Quantum Computing Center (NQCC) and IBM under NQCC's Quantum Computing Access Programme (QCAP).
The authors would like to thank Oleksandr Kyriikenko, Subhayan Roy Moulik, Jonte Hance, 
Simon Williams, and Sergii Strelchuk for productive discussions related to this work.
G.V.G, D.B. and B.C. would like to thank STFC (UK) consolidated grant ST/X000583/1, and G.V.G and B.C. would further like to thank STFC (UK) grant ST/W006251/1 and EPSRC grant EP/W032635/1.
D.B. would like to acknowledge continued support from the Alexander von Humboldt Foundation 
(Germany) in the context of the research fellowship for experienced researchers.

\bibliography{references}

@article{Bauer2016,
  title = {Hybrid Quantum-Classical Approach to Correlated Materials},
  author = {Bauer, Bela and Wecker, Dave and Millis, Andrew J. and Hastings, Matthew B. and Troyer, Matthias},
  journal = {Phys. Rev. X},
  volume = {6},
  issue = {3},
  pages = {031045},
  numpages = {11},
  year = {2016},
  month = {Sep},
  publisher = {American Physical Society},
  doi = {10.1103/PhysRevX.6.031045},
  url = {https://link.aps.org/doi/10.1103/PhysRevX.6.031045}
}

@article{Dumitrescu2018,
  title = {Cloud Quantum Computing of an Atomic Nucleus},
  author = {Dumitrescu, E. F. and McCaskey, A. J. and Hagen, G. and Jansen, G. R. and Morris, T. D. and Papenbrock, T. and Pooser, R. C. and Dean, D. J. and Lougovski, P.},
  journal = {Phys. Rev. Lett.},
  volume = {120},
  issue = {21},
  pages = {210501},
  numpages = {6},
  year = {2018},
  month = {May},
  publisher = {American Physical Society},
  doi = {10.1103/PhysRevLett.120.210501},
  url = {https://link.aps.org/doi/10.1103/PhysRevLett.120.210501}
}

@article{Motta2019,
   title={Determining eigenstates and thermal states on a quantum computer using quantum imaginary time evolution},
   volume={16},
   ISSN={1745-2481},
   url={http://dx.doi.org/10.1038/s41567-019-0704-4},
   DOI={10.1038/s41567-019-0704-4},
   number={2},
   journal={Nature Physics},
   publisher={Springer Science and Business Media LLC},
   author={Motta, Mario and Sun, Chong and Tan, Adrian T. K. and O’Rourke, Matthew J. and Ye, Erika and Minnich, Austin J. and Brandão, Fernando G. S. L. and Chan, Garnet Kin-Lic},
   year={2019},
   month=Nov, pages={205–210} }

@article{Farrell2024,
  title = {Scalable Circuits for Preparing Ground States on Digital Quantum Computers: The Schwinger Model Vacuum on 100 Qubits},
  author = {Farrell, Roland C. and Illa, Marc and Ciavarella, Anthony N. and Savage, Martin J.},
  journal = {PRX Quantum},
  volume = {5},
  issue = {2},
  pages = {020315},
  numpages = {32},
  year = {2024},
  month = {Apr},
  publisher = {American Physical Society},
  doi = {10.1103/PRXQuantum.5.020315},
  url = {https://link.aps.org/doi/10.1103/PRXQuantum.5.020315}
}

@article{Peruzzo2013,
    author = "Peruzzo, Alberto and McClean, Jarrod and Shadbolt, Peter and Yung, Man-Hong and Zhou, Xiao-Qi and Love, Peter J. and Aspuru-Guzik, Al{\'a}n and O'Brien, Jeremy L.",
    title = "{A variational eigenvalue solver on a photonic quantum processor}",
    eprint = "1304.3061",
    archivePrefix = "arXiv",
    primaryClass = "quant-ph",
    doi = "10.1038/ncomms5213",
    journal = "Nature Commun.",
    volume = "5",
    number = "1",
    pages = "4213",
    year = "2014"
}

@article{Maiti2024,
    author = "Maiti, Sandip and Banerjee, Debasish and Chakraborty, Bipasha and Huffman, Emilie",
    title = "{Spontaneous symmetry breaking in a SO(3) non-Abelian lattice gauge theory in 2+1D with quantum algorithms}",
    eprint = "2409.07108",
    archivePrefix = "arXiv",
    primaryClass = "hep-lat",
    doi = "10.1103/PhysRevResearch.7.013283",
    journal = "Phys. Rev. Res.",
    volume = "7",
    number = "1",
    pages = "013283",
    year = "2025"
}

@article{Stetcu2021,
    author = "Stetcu, I. and Baroni, A. and Carlson, J.",
    title = "{Variational approaches to constructing the many-body nuclear ground state for quantum computing}",
    eprint = "2110.06098",
    archivePrefix = "arXiv",
    primaryClass = "nucl-th",
    reportNumber = "LA-UR-21-29364",
    doi = "10.1103/PhysRevC.105.064308",
    journal = "Phys. Rev. C",
    volume = "105",
    number = "6",
    pages = "064308",
    year = "2022"
}

@article{Ciavarella2021,
    author = "Ciavarella, Anthony N. and Chernyshev, Ivan A.",
    title = "{Preparation of the SU(3) lattice Yang-Mills vacuum with variational quantum methods}",
    eprint = "2112.09083",
    archivePrefix = "arXiv",
    primaryClass = "quant-ph",
    reportNumber = "IQuS@UW-21-017",
    doi = "10.1103/PhysRevD.105.074504",
    journal = "Phys. Rev. D",
    volume = "105",
    number = "7",
    pages = "074504",
    year = "2022"
}

@article{Rosanowski:2025fhr,
    author = "Rosanowski, Emil Otis and Eisinger, Jurek and Funcke, Lena and Poschinger, Ulrich and Schmidt-Kaler, Ferdinand",
    title = "{Sample-Based Krylov Quantum Diagonalization for the Schwinger Model on Trapped-Ion and Superconducting Quantum Processors}",
    eprint = "2510.26951",
    archivePrefix = "arXiv",
    primaryClass = "quant-ph",
    month = "10",
    year = "2025"
}

@article{Schuster2023,
    author = {Schuster, Stephan and K{\"u}hn, Stefan and Funcke, Lena and Hartung, Tobias and Pleinert, Marc-Oliver and von Zanthier, Joachim and Jansen, Karl},
    title = "{Studying the phase diagram of the three-flavor Schwinger model in the presence of a chemical potential with measurement- and gate-based quantum computing}",
    eprint = "2311.14825",
    archivePrefix = "arXiv",
    primaryClass = "hep-lat",
    doi = "10.1103/PhysRevD.109.114508",
    journal = "Phys. Rev. D",
    volume = "109",
    number = "11",
    pages = "114508",
    year = "2024"
}

@article{Lumia2021,
    author = "Lumia, Luca and Torta, Pietro and Mbeng, Glen B. and Santoro, Giuseppe E. and Ercolessi, Elisa and Burrello, Michele and Wauters, Matteo M.",
    title = "{Two-Dimensional Z2 Lattice Gauge Theory on a Near-Term Quantum Simulator: Variational Quantum Optimization, Confinement, and Topological Order}",
    eprint = "2112.11787",
    archivePrefix = "arXiv",
    primaryClass = "quant-ph",
    reportNumber = "QDEV CMT NBI 2021",
    doi = "10.1103/PRXQuantum.3.020320",
    journal = "PRX Quantum",
    volume = "3",
    number = "2",
    pages = "020320",
    year = "2022"
}

@article{Chakraborty2020,
    author = "Chakraborty, Bipasha and Honda, Masazumi and Izubuchi, Taku and Kikuchi, Yuta and Tomiya, Akio",
    title = "{Classically emulated digital quantum simulation of the Schwinger model with a topological term via adiabatic state preparation}",
    eprint = "2001.00485",
    archivePrefix = "arXiv",
    primaryClass = "hep-lat",
    doi = "10.1103/PhysRevD.105.094503",
    journal = "Phys. Rev. D",
    volume = "105",
    number = "9",
    pages = "094503",
    year = "2022"
}

@article{McClean2016,
    author = "McClean, Jarrod R. and Kimchi-Schwartz, Mollie E. and Carter, Jonathan and de Jong, Wibe A.",
    title = "{Hybrid quantum-classical hierarchy for mitigation of decoherence and determination of excited states}",
    eprint = "1603.05681",
    archivePrefix = "arXiv",
    primaryClass = "quant-ph",
    doi = "10.1103/PhysRevA.95.042308",
    journal = "Phys. Rev. A",
    volume = "95",
    number = "4",
    pages = "042308",
    year = "2017"
}

@article{Higgott2018,
    author = "Higgott, Oscar and Wang, Daochen and Brierley, Stephen",
    title = "{Variational Quantum Computation of Excited States}",
    eprint = "1805.08138",
    archivePrefix = "arXiv",
    primaryClass = "quant-ph",
    doi = "10.22331/q-2019-07-01-156",
    journal = "Quantum",
    volume = "3",
    pages = "156",
    year = "2019"
}

@article{Nakanishi2019,
    author = "Nakanishi, Ken M. and Mitarai, Kosuke and Fujii, Keisuke",
    title = "{Subspace-search variational quantum eigensolver for excited states}",
    eprint = "1810.09434",
    archivePrefix = "arXiv",
    primaryClass = "quant-ph",
    doi = "10.1103/PhysRevResearch.1.033062",
    journal = "Phys. Rev. Res.",
    volume = "1",
    number = "3",
    pages = "033062",
    year = "2019"
}

@article{Atas2021,
    author = "Atas, Yasar Y. and Zhang, Jinglei and Lewis, Randy and Jahanpour, Amin and Haase, Jan F. and Muschik, Christine A.",
    title = "{SU(2) hadrons on a quantum computer via a variational approach}",
    eprint = "2102.08920",
    archivePrefix = "arXiv",
    primaryClass = "quant-ph",
    doi = "10.1038/s41467-021-26825-4",
    journal = "Nature Commun.",
    volume = "12",
    number = "1",
    pages = "6499",
    year = "2021"
}

@article{Guo2024,
    author = {Guo, Yibin and Angelides, Takis and Jansen, Karl and K{\"u}hn, Stefan},
    title = "{Concurrent VQE for Simulating Excited States of the Schwinger Model}",
    eprint = "2407.15629",
    archivePrefix = "arXiv",
    primaryClass = "quant-ph",
    month = "7",
    year = "2024"
}

@article{Smith2019,
    author = "Smith, Adam and Kim, M. S. and Pollmann, Frank and Knolle, Johannes",
    title = "{Simulating quantum many-body dynamics on a current digital quantum computer}",
    eprint = "1906.06343",
    archivePrefix = "arXiv",
    primaryClass = "quant-ph",
    doi = "10.1038/s41534-019-0217-0",
    journal = "npj Quantum Inf.",
    volume = "5",
    pages = "106",
    year = "2019"
}

@article{Roggero2019,
    author = "Roggero, Alessandro and Li, Andy C. Y. and Carlson, Joseph and Gupta, Rajan and Perdue, Gabriel N.",
    title = "{Quantum Computing for Neutrino-Nucleus Scattering}",
    eprint = "1911.06368",
    archivePrefix = "arXiv",
    primaryClass = "quant-ph",
    reportNumber = "LA-UR-19-31323, INT-PUB-19-052, FERMILAB-PUB-19-547-QIS",
    doi = "10.1103/PhysRevD.101.074038",
    journal = "Phys. Rev. D",
    volume = "101",
    number = "7",
    pages = "074038",
    year = "2020"
}

@article{Martinez2016,
    author = "Martinez, E. A. and others",
    title = "{Real-time dynamics of lattice gauge theories with a few-qubit quantum computer}",
    eprint = "1605.04570",
    archivePrefix = "arXiv",
    primaryClass = "quant-ph",
    doi = "10.1038/nature18318",
    journal = "Nature",
    volume = "534",
    pages = "516--519",
    year = "2016"
}

@article{deJong2021,
    author = "de Jong, Wibe A. and Lee, Kyle and Mulligan, James and P{\l}osko{\'n}, Mateusz and Ringer, Felix and Yao, Xiaojun",
    title = "{Quantum simulation of nonequilibrium dynamics and thermalization in the Schwinger model}",
    eprint = "2106.08394",
    archivePrefix = "arXiv",
    primaryClass = "quant-ph",
    reportNumber = "MIT-CTP/5308",
    doi = "10.1103/PhysRevD.106.054508",
    journal = "Phys. Rev. D",
    volume = "106",
    number = "5",
    pages = "054508",
    year = "2022"
}

@article{DAlessio2015,
    author = "D'Alessio, Luca and Kafri, Yariv and Polkovnikov, Anatoli and Rigol, Marcos",
    title = "{From quantum chaos and eigenstate thermalization to statistical mechanics and thermodynamics}",
    eprint = "1509.06411",
    archivePrefix = "arXiv",
    primaryClass = "cond-mat.stat-mech",
    doi = "10.1080/00018732.2016.1198134",
    journal = "Adv. Phys.",
    volume = "65",
    number = "3",
    pages = "239--362",
    year = "2016"
}

@article{Huffman2021,
    author = "Huffman, Emilie and Garc{\'\i}a Vera, Miguel and Banerjee, Debasish",
    title = "{Toward the real-time evolution of gauge-invariant $\mathbb Z_2$ and $U(1)$ quantum link models on noisy intermediate-scale quantum hardware with error mitigation}",
    eprint = "2109.15065",
    archivePrefix = "arXiv",
    primaryClass = "quant-ph",
    doi = "10.1103/PhysRevD.106.094502",
    journal = "Phys. Rev. D",
    volume = "106",
    number = "9",
    pages = "094502",
    year = "2022"
}

@article{Nagano2023,
    author = "Nagano, Lento and Bapat, Aniruddha and Bauer, Christian W.",
    title = "{Quench dynamics of the Schwinger model via variational quantum algorithms}",
    eprint = "2302.10933",
    archivePrefix = "arXiv",
    primaryClass = "hep-ph",
    doi = "10.1103/PhysRevD.108.034501",
    journal = "Phys. Rev. D",
    volume = "108",
    number = "3",
    pages = "034501",
    year = "2023"
}

@article{Hall2021,
    author = "Hall, Benjamin and Roggero, Alessandro and Baroni, Alessandro and Carlson, Joseph",
    title = "{Simulation of collective neutrino oscillations on a quantum computer}",
    eprint = "2102.12556",
    archivePrefix = "arXiv",
    primaryClass = "quant-ph",
    doi = "10.1103/PhysRevD.104.063009",
    journal = "Phys. Rev. D",
    volume = "104",
    number = "6",
    pages = "063009",
    year = "2021"
}

@article{Ciavarella2020,
    author = "Ciavarella, Anthony",
    title = "{Algorithm for quantum computation of particle decays}",
    eprint = "2007.04447",
    archivePrefix = "arXiv",
    primaryClass = "hep-th",
    reportNumber = "INT-PUB-20-027",
    doi = "10.1103/PhysRevD.102.094505",
    journal = "Phys. Rev. D",
    volume = "102",
    number = "9",
    pages = "094505",
    year = "2020"
}

@article{Oganesyan2007,
    author = "Oganesyan, Vadim and Huse, David A.",
    title = "{Localization of interacting fermions at high temperature}",
    doi = "10.1103/PhysRevB.75.155111",
    journal = "Phys. Rev. B",
    volume = "75",
    number = "15",
    pages = "155111",
    year = "2007"
}

@article{Atas2013,
    author = "Atas, Y. Y. and Bogomolny, E. and Giraud, O. and Roux, G.",
    title = "{Distribution of the Ratio of Consecutive Level Spacings in Random Matrix Ensembles}",
    doi = "10.1103/PhysRevLett.110.084101",
    journal = "Phys. Rev. Lett.",
    volume = "110",
    number = "8",
    pages = "084101",
    year = "2013"
}

@article{Pal2010,
    author = "Pal, Arijeet and Huse, David A.",
    title = "{Many-body localization phase transition}",
    eprint = "1003.2613",
    archivePrefix = "arXiv",
    primaryClass = "cond-mat.dis-nn",
    doi = "10.1103/PhysRevB.82.174411",
    journal = "Phys. Rev. B",
    volume = "82",
    number = "17",
    pages = "174411",
    year = "2010"
}

@article{Turner2017,
    author = "Turner, Christopher J. and Michailidis, Alexios A. and Abanin, Dmitry A. and Serbyn, Maksym and Papic, Zlatko",
    title = "{Weak ergodicity breaking from quantum many-body scars}",
    eprint = "1711.03528",
    archivePrefix = "arXiv",
    primaryClass = "quant-ph",
    doi = "10.1038/s41567-018-0137-5",
    journal = "Nature Phys.",
    volume = "14",
    pages = "745--749",
    year = "2018"
}

@article{Ferris2023,
    author = "Ferris, Kaelyn J. and Wang, Zihang and Hen, Itay and Kalev, Amir and Bronn, Nicholas T. and Vlcek, Vojtech",
    title = "{Exploiting Maximally Mixed States for Spectral Estimation by Time Evolution}",
    eprint = "2312.00687",
    archivePrefix = "arXiv",
    primaryClass = "quant-ph",
    month = "12",
    year = "2023"
}

@article{Langfeld2016,
  title={An efficient algorithm for numerical computations of continuous densities of states},
  author={Langfeld, Kurt and Lucini, Biagio and Pellegrini, Roberto and Rago, Antonio},
  journal={The European Physical Journal C},
  volume={76},
  number={6},
  pages={306},
  year={2016},
  publisher={Springer}
}

@article{Kitaev1995,
  title={Quantum measurements and the Abelian stabilizer problem},
  author={Kitaev, A Yu},
  journal={arXiv preprint quant-ph/9511026},
  year={1995}
}

@article{Somma2002,
  title = {Simulating physical phenomena by quantum networks},
  author = {Somma, R. and Ortiz, G. and Gubernatis, J. E. and Knill, E. and Laflamme, R.},
  journal = {Phys. Rev. A},
  volume = {65},
  issue = {4},
  pages = {042323},
  numpages = {17},
  year = {2002},
  month = {Apr},
  publisher = {American Physical Society},
  doi = {10.1103/PhysRevA.65.042323},
  url = {https://link.aps.org/doi/10.1103/PhysRevA.65.042323}
}

@article{Somma2019,
    author = "Somma, Rolando D.",
    title = "{Quantum eigenvalue estimation via time series analysis}",
    doi = "10.1088/1367-2630/ab5c60",
    journal = "New J. Phys.",
    volume = "21",
    number = "12",
    pages = "123025",
    year = "2019"
}

@article{Stenger2022,
  title = {Simulating spectroscopy experiments with a superconducting quantum computer},
  author = {Stenger, John P. T. and Ben-Shach, Gilad and Pekker, David and Bronn, Nicholas T.},
  journal = {Phys. Rev. Res.},
  volume = {4},
  issue = {4},
  pages = {043106},
  numpages = {10},
  year = {2022},
  month = {Nov},
  publisher = {American Physical Society},
  doi = {10.1103/PhysRevResearch.4.043106},
  url = {https://link.aps.org/doi/10.1103/PhysRevResearch.4.043106}
}

@article{Poulin2018,
    author = "Poulin, David and Kitaev, Alexei and Steiger, Damian S. and Hastings, Matthew B. and Troyer, Matthias",
    title = "{Quantum Algorithm for Spectral Measurement with a Lower Gate Count}",
    doi = "10.1103/PhysRevLett.121.010501",
    journal = "Phys. Rev. Lett.",
    volume = "121",
    number = "1",
    pages = "010501",
    year = "2018"
}

@article{Lee2022,
    author = "Lee, Woo-Ram and Scott, Ryan and Scarola, V. W.",
    title = "{Hybrid quantum-gap-estimation algorithm using a filtered time series}",
    eprint = "2212.14039",
    archivePrefix = "arXiv",
    primaryClass = "quant-ph",
    doi = "10.1103/PhysRevA.109.052403",
    journal = "Phys. Rev. A",
    volume = "109",
    number = "5",
    pages = "052403",
    year = "2024"
}

@article{trotter1959product,
  author  = {Trotter, H. F.},
  title   = {On the Product of Semi-Groups of Operators},
  journal = {Proceedings of the American Mathematical Society},
  volume  = {10},
  number  = {4},
  pages   = {545--551},
  year    = {1959},
  month   = {aug},
  doi     = {10.1090/S0002-9939-1959-0108732-6},
  url     = {https://www.jstor.org/stable/2033649},
  publisher = {American Mathematical Society}
}

@article{Bruschi2024,
    author = "Bruschi, Matteo and Gallina, Federico and Fresch, Barbara",
    title = "{A Quantum Algorithm from Response Theory: Digital Quantum Simulation of Two-Dimensional Electronic Spectroscopy}",
    doi = "10.1021/acs.jpclett.3c03499",
    journal = "J. Phys. Chem. Lett.",
    volume = "15",
    number = "5",
    pages = "1484--1492",
    year = "2024"
}

@article{Lee2024,
    author = "Lee, Woo-Ram and Myers, Nathan M. and Scarola, V. W.",
    title = "{Noise-resilient and resource-efficient hybrid algorithm for robust quantum gap estimation}",
    eprint = "2405.10306",
    archivePrefix = "arXiv",
    primaryClass = "quant-ph",
    doi = "10.1103/PhysRevA.111.022421",
    journal = "Phys. Rev. A",
    volume = "111",
    number = "2",
    pages = "022421",
    year = "2025"
}

@article{VilchezEstevez2025,
    author = "Vilchez-Estevez, Lucia and Santos, Raul A. and Wang, Sabrina and Gambetta, Filippo Maria",
    title = "{Extracting the spin excitation spectrum of a fermionic system using a quantum processor}",
    eprint = "2501.04649",
    archivePrefix = "arXiv",
    primaryClass = "quant-ph",
    month = "1",
    year = "2025"
}

@article{Weisse2006,
  author  = {Wei{\ss}e, Alexander and Wellein, Gerhard and Alvermann, Andreas and Fehske, Holger},
  title   = {The Kernel Polynomial Method},
  journal = {Rev. Mod. Phys.},
  volume  = {78},
  pages   = {275--306},
  year    = {2006},
  doi     = {10.1103/RevModPhys.78.275},
  eprint  = {cond-mat/0504627},
  archivePrefix = {arXiv}
}

@article{Parrish2019,
  author  = {Parrish, Robert M. and McMahon, Peter L.},
  title   = {Quantum Filter Diagonalization: Quantum Eigendecomposition without Full Quantum Phase Estimation},
  year    = {2019},
  eprint  = {1909.08925},
  archivePrefix = {arXiv},
  primaryClass = {quant-ph}
}

@article{Shen2022,
  author  = {Shen, Yizhi and Klymko, Katherine and Sud, James and Williams-Young, David B. and de Jong, Wibe A. and Tubman, Norm M.},
  title   = {Real-Time Krylov Theory for Quantum Computing Algorithms},
  journal = {Quantum},
  volume  = {7},
  pages   = {1066},
  year    = {2023},
  doi     = {10.22331/q-2023-07-25-1066},
  eprint  = {2208.01063},
  archivePrefix = {arXiv},
  primaryClass = {quant-ph}
}

@article{Summer2023,
  author  = {Summer, Alessandro and Chiaracane, Cecilia and Mitchison, Mark T. and Goold, John},
  title   = {Calculating the Many-Body Density of States on a Digital Quantum Computer},
  journal = {Phys. Rev. Research},
  volume  = {6},
  pages   = {013106},
  year    = {2024},
  doi     = {10.1103/PhysRevResearch.6.013106},
  eprint  = {2303.13476},
  archivePrefix = {arXiv},
  primaryClass = {quant-ph}
}

@article{Rico2018,
title = {SO(3) “Nuclear Physics” with ultracold Gases},
journal = {Annals of Physics},
volume = {393},
pages = {466-483},
year = {2018},
issn = {0003-4916},
doi = {https://doi.org/10.1016/j.aop.2018.03.020},
url = {https://www.sciencedirect.com/science/article/pii/S0003491618300757},
author = {E. Rico and M. Dalmonte and P. Zoller and D. Banerjee and M. Bögli and P. Stebler and U.-J. Wiese}
}

@article{Xu2006,
  author = {Xu, Jun and Ma, H.-R.},
  title = {Density of States of a Two-Dimensional {XY} Model from the Wang-Landau Algorithm},
  journal = {Physical Review E},
  volume = {75},
  number = {4},
  pages = {041115},
  year = {2007},
  doi = {10.1103/PhysRevE.75.041115},
  eprint = {cond-mat/0611039},
  archivePrefix = {arXiv}
}

@article{keating2015spectra,
  author  = {Keating, J. P. and Linden, N. and Wells, H. J.},
  title   = {Spectra and Eigenstates of Spin Chain Hamiltonians},
  journal = {Communications in Mathematical Physics},
  volume  = {338},
  number  = {1},
  pages   = {81--102},
  year    = {2015},
  doi     = {10.1007/s00220-015-2366-0},
  eprint  = {1403.1121},
  archivePrefix = {arXiv}
}

@article{mon1975statistical,
  author  = {Mon, K. K. and French, J. B.},
  title   = {Statistical Properties of Many-Particle Spectra},
  journal = {Annals of Physics},
  volume  = {95},
  number  = {1},
  pages   = {90--111},
  year    = {1975}
}

@article{benet2003review,
  author  = {Benet, L. and Weidenm{\"u}ller, H. A.},
  title   = {Review of the $k$-Body Embedded Ensembles of {G}aussian Random Matrices},
  journal = {Journal of Physics A: Mathematical and General},
  volume  = {36},
  number  = {12},
  pages   = {3569--3594},
  year    = {2003},
  doi     = {10.1088/0305-4470/36/12/340},
  eprint  = {cond-mat/0207656},
  archivePrefix = {arXiv}
}

@book{kota2014embedded,
  author    = {Kota, V. K. B.},
  title     = {Embedded Random Matrix Ensembles in Quantum Physics},
  series    = {Lecture Notes in Physics},
  volume    = {884},
  publisher = {Springer},
  address   = {Heidelberg},
  year      = {2014}
}

@article{kempe2006complexity,
  author  = {Kempe, Julia and Kitaev, Alexei and Regev, Oded},
  title   = {The Complexity of the Local {H}amiltonian Problem},
  journal = {SIAM Journal on Computing},
  volume  = {35},
  number  = {5},
  pages   = {1070--1097},
  year    = {2006},
  doi     = {10.1137/S0097539704445226}
}

@book{macwilliams1977theory,
  author    = {MacWilliams, F. J. and Sloane, N. J. A.},
  title     = {The Theory of Error-Correcting Codes},
  publisher = {North-Holland},
  address   = {Amsterdam},
  year      = {1977}
}

@article{Sinha2009,
  author = {Sinha, Suman and Roy, Soumen Kumar},
  title = {Performance of Wang-Landau Algorithm in Continuous Spin Models and a Case Study: Modified {XY}-Model},
  journal = {Physics Letters A},
  volume = {373},
  number = {3},
  pages = {308--314},
  year = {2009},
  doi = {10.1016/j.physleta.2008.11.034},
  eprint = {0711.1031},
  archivePrefix = {arXiv},
  primaryClass = {cond-mat.stat-mech}
}

@article{Langfeld2012,
  author = {Langfeld, Kurt and Lucini, Biagio and Rago, Antonio},
  title = {The Density of States in Gauge Theories},
  journal = {Physical Review Letters},
  volume = {109},
  number = {11},
  pages = {111601},
  year = {2012},
  doi = {10.1103/PhysRevLett.109.111601},
  eprint = {1204.3243},
  archivePrefix = {arXiv},
  primaryClass = {hep-lat}
}

@article{Chamberland2017,
    author = "Chamberland, Christopher and Iyer, Pavithran and Poulin, David",
    title = "{Fault-tolerant quantum computing in the Pauli or Clifford frame with slow error diagnostics}",
    eprint = "1704.06662",
    archivePrefix = "arXiv",
    primaryClass = "quant-ph",
    doi = "10.22331/q-2018-01-04-43",
    journal = "Quantum",
    volume = "2",
    pages = "43",
    year = "2018"
}

@article{Cardy1985,
  author  = {Cardy, John L.},
  title   = {Universal amplitudes in finite-size scaling: generalisation to arbitrary dimensionality},
  journal = {J. Phys. A: Math. Gen.},
  volume  = {18},
  number  = {13},
  pages   = {L757},
  year    = {1985},
  doi     = {10.1088/0305-4470/18/13/005},
}

@article{Zhu2023,
  author  = {Zhu, W. and Han, Chao and Huffman, Emilie and Hofmann, Johannes S. and He, Yin-Chen},
  title   = {Uncovering Conformal Symmetry in the {3D} {I}sing Transition: State-Operator Correspondence from a Quantum Fuzzy Sphere Regularization},
  journal = {Phys. Rev. X},
  volume  = {13},
  number  = {2},
  pages   = {021009},
  year    = {2023},
  doi     = {10.1103/PhysRevX.13.021009},
  eprint  = {2210.13482},
  archivePrefix = {arXiv},
  primaryClass = {cond-mat.stat-mech},
}

@article{SUZUKI1990319,
title = {Fractal decomposition of exponential operators with applications to many-body theories and Monte Carlo simulations},
journal = {Physics Letters A},
volume = {146},
number = {6},
pages = {319-323},
year = {1990},
issn = {0375-9601},
doi = {https://doi.org/10.1016/0375-9601(90)90962-N},
url = {https://www.sciencedirect.com/science/article/pii/037596019090962N},
author = {Masuo Suzuki}
}

@article{Schmitz2023,
    author = "Schmitz, Albert T. and Sawaya, Nicolas P. D. and Johri, Sonika and Matsuura, A. Y.",
    title = "{Graph optimization perspective for low-depth Trotter-Suzuki decomposition}",
    eprint = "2103.08602",
    archivePrefix = "arXiv",
    primaryClass = "quant-ph",
    doi = "10.1103/PhysRevA.109.042418",
    journal = "Phys. Rev. A",
    volume = "109",
    number = "4",
    pages = "042418",
    year = "2024"
}

@inproceedings{Paykin2023,
    author = "Paykin, Jennifer and Schmitz, Albert T. and Ibrahim, Mohannad and Wu, Xin-Chuan and Matsuura, A. Y.",
    title = "{PCOAST: A Pauli-Based Quantum Circuit Optimization Framework}",
    eprint = "2305.10966",
    archivePrefix = "arXiv",
    primaryClass = "quant-ph",
    booktitle = "{2023 IEEE International Conference on Quantum Computing and Engineering (QCE)}",
    pages = "715--726",
    doi = "10.1109/QCE57702.2023.00087",
    year = "2023"
}

@inproceedings{Liu2024,
    author = "Liu, Ji and Gonzales, Alvin and Huang, Benchen and Saleem, Zain Hamid and Hovland, Paul",
    title = "{QuCLEAR: Clifford Extraction and Absorption for Quantum Circuit Optimization}",
    eprint = "2408.13316",
    archivePrefix = "arXiv",
    primaryClass = "quant-ph",
    booktitle = "{2025 IEEE International Symposium on High Performance Computer Architecture (HPCA)}",
    pages = "158--172",
    doi = "10.1109/HPCA61900.2025.00023",
    year = "2025"
}

@article{VanGoffrier2024,
       author = {{Van Goffrier}, Graham and {Banerjee}, Debasish and {Chakraborty}, Bipasha and {Huffman}, Emilie and {Maiti}, Sandip},
        title = "{Towards the phase diagram of fermions coupled with $SO(3)$ quantum links in $(2+1)$-D}",
      journal = {arXiv e-prints},
         year = 2024,
        month = dec,
          eid = {arXiv:2412.09691},
        pages = {arXiv:2412.09691},
          doi = {10.48550/arXiv.2412.09691},
archivePrefix = {arXiv},
       eprint = {2412.09691},
 primaryClass = {hep-lat},
       adsurl = {https://ui.adsabs.harvard.edu/abs/2024arXiv241209691V}
}
\newpage
\appendix
 \section{Other quantum spectral estimation strategies} \label{sec:otherSpecStrageies}
 The quantum eigenvalue estimation problem (QEEP) introduced in Ref.~\cite{Somma2019} 
provided a more careful formulation of the same task. Using overlapping bins and smooth
window functions are able to reduce the long-time cutoff needed to control truncation errors
and avoid other instabilities. The result is a finite-resolution spectral distribution with 
rigorous error and confidence guarantees, obtained from one-ancilla time-series measurements. 
Alternative classical post-processing can help further when the signal contains only a small 
number of well-separated frequencies. 

 The same time-domain structure appears in algorithms which target spectral gaps or response 
functions rather than the density of states. In probe-qubit spectroscopy, an auxiliary 
two-level system is weakly coupled to the model and its transition probability is measured 
while the probe frequency is scanned \cite{Stenger2022}. Resonances identify energy differences 
which are connected by the coupling operator. Quench protocols instead prepare a perturbed state 
and reconstruct retarded correlation functions from the subsequent dynamics; recent implementations 
have extracted excitation spectra of interacting fermionic systems on quantum processors 
\cite{VilchezEstevez2025}. Filtered time-series methods can focus directly on a spectral gap 
and can be more resilient to state-preparation and hardware noise than a reconstruction of 
the entire spectrum \cite{Lee2022,Lee2024}. Another method \cite{Bruschi2024} takes inspiration 
from nonlinear optical response theory, computing dipole-operator response functions from a 
Hadamard test on a single-controlled time evolution. These methods are physically appealing because 
they reproduce the logic of an experiment, but the observed lines are weighted by transition 
matrix elements and by the choice of initial state. A missing line need not imply the absence 
of an eigenstate; it may simply be dark to the chosen probe.

 Moment methods provide another route to a smooth density of states. After shifting and rescaling 
the Hamiltonian so that the spectrum of $\widetilde H$ lies in $[-1,1]$, the normalised density 
can be expanded in Chebyshev polynomials with coefficients,
\begin{equation} \label{eq:ChebyshevMom}
    \mu_m = \frac{1}{D} \Tr\!\left[T_m(\widetilde H)\right].
\end{equation}
 Classically, the kernel polynomial method evaluates these moments recursively and damps the 
 truncated expansion to suppress Gibbs oscillations \cite{Weisse2006}. Quantum versions estimate 
 the moments using controlled evolutions together with stochastic trace evaluation, and have been 
 demonstrated for many-body densities of states on trapped-ion hardware \cite{Summer2023}. 
 The moment and time-domain pictures are closely related: both replace explicit diagonalization 
 by the estimation of a compressed representation of the spectrum. Their practical difference 
 lies in the basis used for the reconstruction, the required circuit primitives and the manner 
 in which finite resolution is regularised.

\section{Pauli-Frame Optimization for Shallow Trotter Step}\label{sec:pauliframe}

A Pauli frame is, for our purposes, a bookkeeping method for how Clifford operations conjugate Pauli operators. It is widely used in fault-tolerant quantum computation, where Pauli/Clifford corrections are propagated forward in time rather than applied immediately, thereby mitigating latency and measurement delays \cite{Chamberland2017}. The same underlying mechanism serves a different purpose for us as well as for other authors \cite{Schmitz2023}, in that we are focused on the reordering of one large Clifford subcircuit, and where possible the partial factoring of that circuit into classical post-processing \cite{Paykin2023}. We note that the application of Pauli frames to gauge theory Trotter steps appears to be absent from the literature, and while we plan to repair this more broadly in an upcoming work, in this Appendix the intention is only to improve the depth of the controlled-evolution block of the quantum spectral sampling algorithm.

For clarity of exposition, we follow \cite{Schmitz2023} in first working in Pauli space modulo overall phase.  In this representation, $n$-qubit Pauli operators are identified with binary symplectic vectors in $\mathbb{F}_2^{2n}$, and conjugation by any Clifford unitary acts linearly on this space, i.e. as the symplectic automorphisms $\mathrm{Sp}(2n,\mathbb{F}_2)$ which preserve commutation relations. As discussed in \cite{Schmitz2023}, incorporating the sign bit directly into this vector space destroys non-degeneracy of the bilinear form and breaks strict linearity for certain two-qubit Cliffords.  It is therefore most straightforward to separate support updates from phase updates, which are restored only after the desired Clifford action has been constructed and optimized.

\begin{definition}[Pauli frame]
Let $\mathcal{P}_n$ denote the $n$-qubit Pauli group and let $\overline{\mathcal{P}}_n = \mathcal{P}_n / \{\pm 1, \pm i\}$ denote Pauli operators modulo global phase.  
A \emph{Pauli frame} is a classical record $F \in \overline{\mathcal{P}}_n$ such that the physical state satisfies
\begin{equation}
\rho_{\mathrm{phys}} = F \rho_{\mathrm{ideal}} F^\dagger .
\end{equation}

It is often convenient to represent a Clifford-induced frame update by listing the images of the single-qubit Pauli generators in two columns,
\begin{equation}
\begin{pmatrix}
U^\dagger Z_1 U & U^\dagger X_1 U \\
U^\dagger Z_2 U & U^\dagger X_2 U \\
\vdots & \vdots \\
U^\dagger Z_n U & U^\dagger X_n U
\end{pmatrix},
\end{equation}
where each entry is an $n$-qubit Pauli operator written as a Pauli string.  

Upon expressing each Pauli operator modulo phase as a binary symplectic vector $(\mathbf{x}\mid\mathbf{z}) \in \mathbb{F}_2^{2n}$, this two-column display corresponds to a $2n \times 2n$ binary matrix $S_U$ representing conjugation by $U$ on phase-free Pauli space.  For Clifford $U$, the matrix $S_U$ is symplectic and satisfies
\begin{equation}
S_U^T \Omega S_U = \Omega ,
\end{equation}
where $\Omega$ is the standard symplectic form on $\mathbb{F}_2^{2n}$.
\end{definition}

\begin{theorem}[Frame propagation under Clifford gates]
\label{FramePropagationThm}
Let $U$ be a Clifford unitary and suppose
\begin{equation}
\ket{\psi_{\mathrm{phys}}}=P\ket{\psi_{\mathrm{ideal}}},\qquad P\in\overline{\mathcal P}_n .
\end{equation}
Then
\begin{equation}
U\ket{\psi_{\mathrm{phys}}}
=
P' U\ket{\psi_{\mathrm{ideal}}},
\qquad
P' = UPU^\dagger \in \overline{\mathcal P}_n .
\end{equation}
\end{theorem}

\begin{proof}
Because the Clifford group normalizes the Pauli group,
$UPU^\dagger$ is again a Pauli operator modulo phase. Hence
\begin{equation}
UP\ket{\psi_{\mathrm{ideal}}}
=
(UPU^\dagger)U\ket{\psi_{\mathrm{ideal}}},
\end{equation}
so Clifford evolution updates the Pauli frame deterministically by conjugation.
In the binary symplectic representation, this is precisely the action of the
symplectic matrix $S_U$ associated with $U$.
\end{proof}

The Clifford group is generated by the single–qubit gates $H$ and $S$ together with the two–qubit gate $\mathrm{CNOT}$.  It is therefore sufficient to understand how these three gates transform Pauli operators under conjugation.  Since any Clifford circuit can be decomposed into this generating set, the induced action of an arbitrary Clifford on Pauli strings follows directly from the rules derived below.
The Hadamard gate exchanges the $X$-columns and $Z$-columns on the qubit to which it is applied, i.e.
\begin{equation}
Z \mapsto H^\dagger Z H = X,~X \mapsto H^\dagger X H = Z,
\end{equation}
while the phase gate performs a $\pi/2$ rotation around the $Z$ axis, i.e. up to a sign,
\begin{equation}
Z \mapsto S^\dagger Z S = Z,~X \mapsto S^\dagger X S = Y,~Y \mapsto S^\dagger Y S = X.
\end{equation}
Conjugation by $\mathrm{CNOT}$ has the effect that the $X$-column operator on the control qubits is multiplied onto that of the target qubit, while the $Z$-column operator on the target is multiplied onto that of the control.  Explicitly,
\begin{align}
X_c &\mapsto \mathrm{CNOT}^\dagger (X_c) \mathrm{CNOT} = X_c X_t, \\
Z_t &\mapsto \mathrm{CNOT}^\dagger (Z_t) \mathrm{CNOT} = Z_c Z_t.
\end{align}

In the Pauli frame representation introduced above, a Clifford unitary $U$ is specified by the images of the generators $Z_j$ and $X_j$. For a given gate sequence this image may be computed from repeated application of the above conjugations.

After the phase-free update rules are understood, we restore signs by fixing a $Y$ convention (e.g.\ $Y=-iZX$) and tracking a phase bit in the tableau. In our circuit derivations, we perform this sign restoration post hoc.

As a toy model preceding the full gauge/QLM construction, we consider how to construct a frame which contains two four-body Pauli terms of interest (e.g.\ $Z_0 Z_1 Z_2 Z_3$ and $X_0 X_1 Z_2 Z_3$), illustrating a compact structure which will prove useful for our full Hamiltonian. A similar term appears in \cite{Schmitz2023}, but without the multi-operator parallelization which we introduce.
{\footnotesize
\setlength{\arraycolsep}{0pt}
\begin{align*}
&\begin{pmatrix}
Z_0 & X_0 \\
Z_1 & X_1 \\
Z_2 & X_2 \\
Z_3 & X_3
\end{pmatrix}
\overset{C^0_1}{\rightarrow}
\begin{pmatrix}
Z_0 & X_0X_1 \\
Z_0Z_1 & X_1 \\
Z_2 & X_2 \\
Z_3 & X_3
\end{pmatrix}
\overset{C^2_3}{\rightarrow}
\begin{pmatrix}
Z_0 & X_0X_1 \\
Z_0Z_1 & X_1 \\
Z_2 & X_2X_3 \\
Z_2Z_3 & X_3
\end{pmatrix}
\overset{H^0}{\rightarrow}
\\[4pt]
&\begin{pmatrix}
X_0X_1 & Z_0 \\
Z_0Z_1 & X_1 \\
Z_2 & X_2X_3 \\
Z_2Z_3 & X_3
\end{pmatrix}
\overset{C^3_0}{\rightarrow}
\begin{pmatrix}
X_0X_1Z_2Z_3 & Z_0 \\
Z_0Z_1 & X_1 \\
Z_2 & X_2X_3 \\
Z_2Z_3 & Z_0X_3
\end{pmatrix}
\overset{C^3_1}{\rightarrow}
\begin{pmatrix}
X_0X_1Z_2Z_3 & Z_0 \\
Z_0Z_1Z_2Z_3 & X_1 \\
Z_2 & X_2X_3 \\
Z_2Z_3 & Z_0X_1X_3
\end{pmatrix}.
\end{align*}
}

A useful perspective is provided by the Pauli Frame Graph \cite{Schmitz2023}, in which nodes represent Pauli frames and edges correspond to elementary Clifford transformations. Circuit synthesis may then be regarded as a walk through this graph, with the cost of moving between frames determined by the Clifford operations required to expose the next Pauli rotation. 

This viewpoint is illustrated schematically in Fig.~\ref{fig:hamming_comparison}. Standard compute-uncompute synthesis treats each Pauli rotation independently, returning to a reference frame before proceeding to the next term. By contrast, a Pauli-frame construction retains the intermediate Clifford frame and moves directly between successive targets. The resulting circuit depth is therefore controlled not only by the individual Pauli weights, but also by the \textit{distance} between consecutive frame configurations. We will see that this permits several multi-qubit rotations to share much of the same Clifford structure rather than repeatedly reconstructing it.

The remaining problem is then to determine which sequence of Pauli targets, and which sequence of Clifford frame transitions between them, yields the shallowest circuit. The automated procedure used to perform that optimization is described in Appendix~\ref{sec:pauliframecompilation}.

\begin{figure}[t!]
    \centering
    \includegraphics[width=0.5\textwidth]{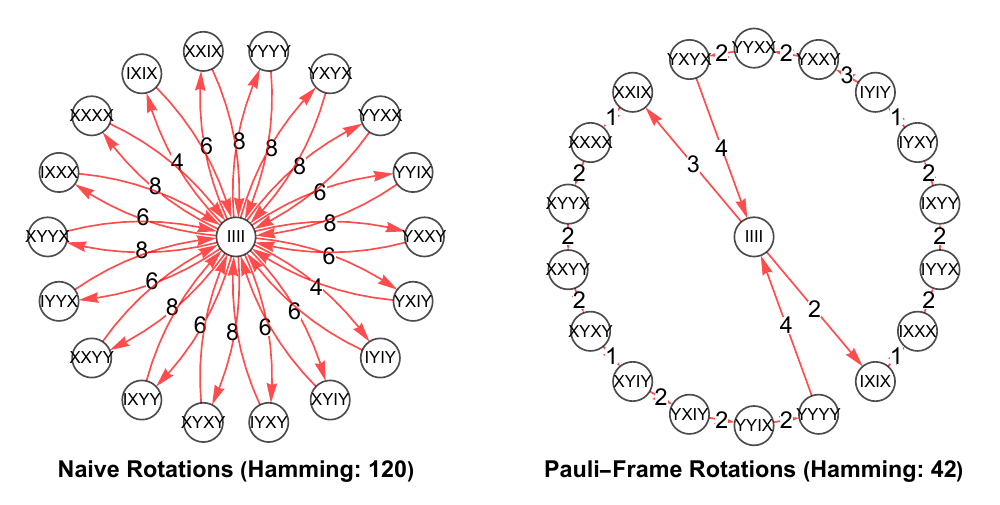}
    \caption{Graphical efficiency comparison between standard compute-uncompute circuitization (left) of a multi-term Trotter step, and two parallel walks of near-optimal Hamming distance through the terms.}
    \label{fig:hamming_comparison}
\end{figure}

\begin{figure*}[t!]
    \centering
    \includegraphics[width=0.95\textwidth]{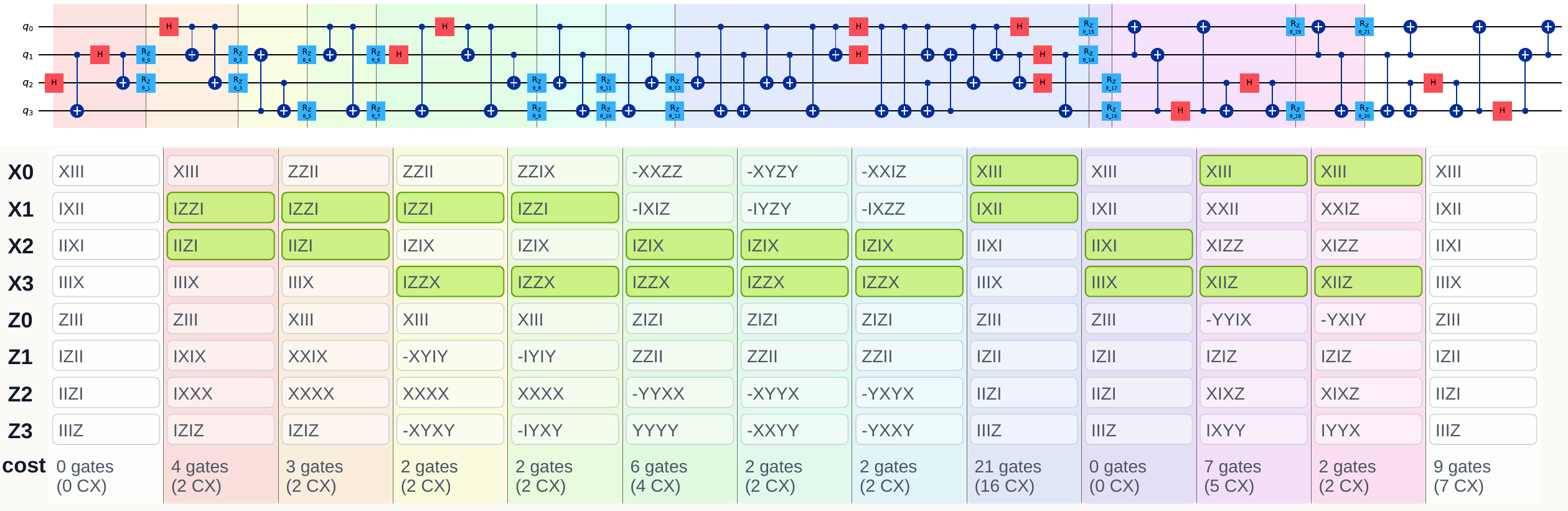}
    \caption{Pauli-frame-optimised Trotter step circuit for the $L=2$ $SO(3)$ QLM Hamiltonian, realising the two parallel computation paths of Figure~\ref{fig:hamming_comparison} (right).}
    \label{fig:trotter_circuit}
\end{figure*}

\section{Automatic Pauli-Frame Circuit Compilation via Recursive Branch-and-Bound Search}\label{sec:pauliframecompilation}

The prior section has clarified what we specifically mean by a multi-target Pauli-frame construction for an $n$-qubit Trotter step. However, so far we have relied upon manual heuristics to identify short circuits which efficiently progress from one multi-target Pauli-frame to the next. This has been remarkably successful for two-targets, enabling us to reduce our circuit depth by a large factor; but is of course not an extendable strategy for $M$-targets, for Hamiltonians with a greater number of terms. More relevantly, manual heuristics are slow, much too slow for the next improvement we would like to include: an enumeration over all possible orderings of Hamiltonian terms in our Trotter step.

A moonshot goal is to (recursively) explore all possible orderings of terms, circuitize them via Pauli-frame compilation, and in so doing, select an ordering/circuitization which minimises both Trotter error and circuit depth. Immediately we can realise that it is almost certainly impossible to minimise both quantities; there will be some Pareto frontier against which the tradeoff between depth and Trotter error is selected. 

Rather than thoroughly explore this frontier, we will factor the problem by first choosing a good-enough improvement of Trotter error, based on a commuting block decomposition. Then within the family of orderings constrained by commuting blocks, we will perform a recursive branch-and-bound search to achieve minimal Pauli-frame circuit depth. Related decompositions into commuting Pauli blocks appear, for example, in QuCLEAR \cite{Liu2024}.

The implementation used here is not a fully general Clifford synthesis algorithm. It is a structured search procedure specialised to Pauli product-formula circuits. Its inputs are $(1)$ a list of real-coefficient Pauli evolution terms and their commuting block structure, and $(2)$ a bounded search policy for finding short Clifford transitions between Pauli frames.

The output is a quantum circuit of the schematic form $U_{\mathrm{PF}}(t) = C_0 D_0(t) C_1 D_1(t) \cdots C_{B-1}D_{B-1}(t) C_{\mathrm{unwind}}$, where each $C_b$ is a time-independent Clifford transition between Pauli frames, each $D_b(t)$ is a diagonal layer of one or two $R_z$ rotations, and $C_{\mathrm{unwind}}$ returns the final Pauli frame to the initial, computational frame.

The Hamiltonian is first expanded as an ordered list of Pauli terms, which we then partition into mutually commuting families by means of a greedy clique-cover procedure; this of course is an upfront step which is unchanged by the model parameters of the Hamiltonian, although it could be further optimised in symmetric regimes where some terms decouple. From this point on, we restrict our attention to orderings which preserve the chosen commuting-family structure. The motivation is that ordering terms within commuting families should not directly increase first-order Trotter error. Nevertheless, a nontrivial circuit optimization problem remains even with this dramatic reduction in search space.

Within each commuting family, terms are grouped into ordered two-term blocks wherever possible. For a family of even size $2m$, we define a \textit{candidate} as an ordered perfect matching $\bigl((i_1,j_1),(i_2,j_2),\ldots,(i_m,j_m)\bigr)$. Putting aside for now how we will select such candidates, we will first specify how a Pauli-frame circuitization is performed for any such candidate.

The allowed elementary Clifford operations are, as usual,
\begin{equation}
H_q,\qquad S_q,\qquad S_q^\dagger,\qquad \mathrm{CNOT}_{q\to r},
\end{equation}
although in our experiments with the $SO(3)$ QLM, we found it beneficial to disable explicit phase gates in the search.

For a one-term block, a deterministic construction is possible. One chooses a pivot qubit from the support of the Pauli string, applies local basis changes $X \mapsto Z \quad \text{using } H$, $Y \mapsto Z \quad \text{using } S^\dagger H$, then applies a fixed CNOT fan-in across the support so that a single $Z$-image equals the target Pauli string.

For a two-term block, if the two Pauli strings happen to commute and enjoy disjoint support, two one-term constructions can be combined and assigned to two distinct pivots.

However, overlapping commuting pairs must be handled by a more expensive bounded-frame search over the graph of Pauli frames generated by elementary Clifford gates. We chose a per-gate weighted cost
\begin{equation}
    \mathrm{cost}(H)=1,
\qquad
\mathrm{cost}(S)=1,
\qquad
\mathrm{cost}(\mathrm{CNOT})=2,
\end{equation}
which is easy to modify based on hardware properties. Two search backends were implemented:

\begin{enumerate}
\item \textbf{Targeted search.} Run a bounded Dijkstra search from the current frame until a frame satisfying the two-Pauli diagonalization constraint is found.
\item \textbf{Reachable-ball search.} Precompute all frames reachable from the current frame within the bound, then scan this bounded ball for frames satisfying different pair constraints.
\end{enumerate}

The reachable-ball backend is only advantageous as an amortization policy when many pair constraints are queried from the same source frame; in this case it is very advantageous indeed. One must also be careful in the case of entry/final pair transitions, since we are optimising one commuting-block at a time and do not want to count the transitions which will change when the blocks are connected together.

For each commuting family, the algorithm recursively searches over ordered pairings as follows:

\begin{enumerate}
\item Start from the identity Pauli frame and the full term set of the family.
\item Enumerate possible next pairs among the remaining terms.
\item Compile the chosen pair from the current frame using the relevant entry/internal search policy.
\item If compilation fails within the cost bound, prune all completions of that prefix.
\item If compilation succeeds, recurse on the remaining terms using the block's final frame.
\item When no terms remain, materialise the candidate circuit and compute its CNOT count, Clifford count, and depth.
\end{enumerate}


Pruning refers to several strategies for rejecting candidates based on prior computations. First, if a pair cannot be exposed from the current frame within the active cost bound, then all candidate orderings beginning with that prefix are skipped. Second, if the same summary state has already been reached with less or equal internal prefix cost, the new prefix is discarded. Third, once a complete candidate has been found, any prefix whose internal cost is already strictly worse than the incumbent internal cost is discarded. Equal-cost prefixes are retained, because they may still win the flat-order tie-break.

This local search produces a selected ordering for each commuting family. A second, global stage then searches over permutations of these families. For each permutation, the first pair of a new family is recompiled as a connector from the frame left by the preceding family, and the complete circuit is materialised and compared using its depth, CNOT count, and total Clifford count. The resulting circuit is therefore the best circuit found within the chosen commuting decomposition and bounded search policy, rather than a proof of global optimality over arbitrary Pauli orderings and Clifford circuits.

For the reader interested in deploying this algorithm for their own purposes, or indeed extending it, we report several false trails which we ultimately rejected:


Manual two-target Pauli-frame constructions were effective for small examples, but did not scale to enumerating many pair and family orderings. They were used as motivation but replaced by automated bounded search.

One early policy used reachable-ball search too aggressively for bounded internal queries. This caused large overhead because many source frames were not reused enough to amortise the cost of building a full bounded neighbourhood. The final policy uses reachable-ball only when cached or when the predicted number of pair queries from a source frame is sufficiently large.

Increasing the internal transition bound admits more pair transitions, but can sharply increase search time. In our tuning runs, larger internal bounds increased the number of accepted candidates without improving the best local objective. We therefore used a smaller internal bound and relied on pruning.

Large entry bounds allowed a few very expensive depth-zero searches. These dominated wall time but did not improve the best circuits found in the tuned runs. 

Figure~\ref{fig:trotter_circuit} shows the Pauli-frame-optimised controlled-Trotter step resulting from the procedure outlined here, for the $L=2$ $SO(3)$ QLM.

\section{Details of Disorder Hardware Runs}\label{sec:disorder_runs}

\begin{figure}
\includegraphics[width=0.49\textwidth]{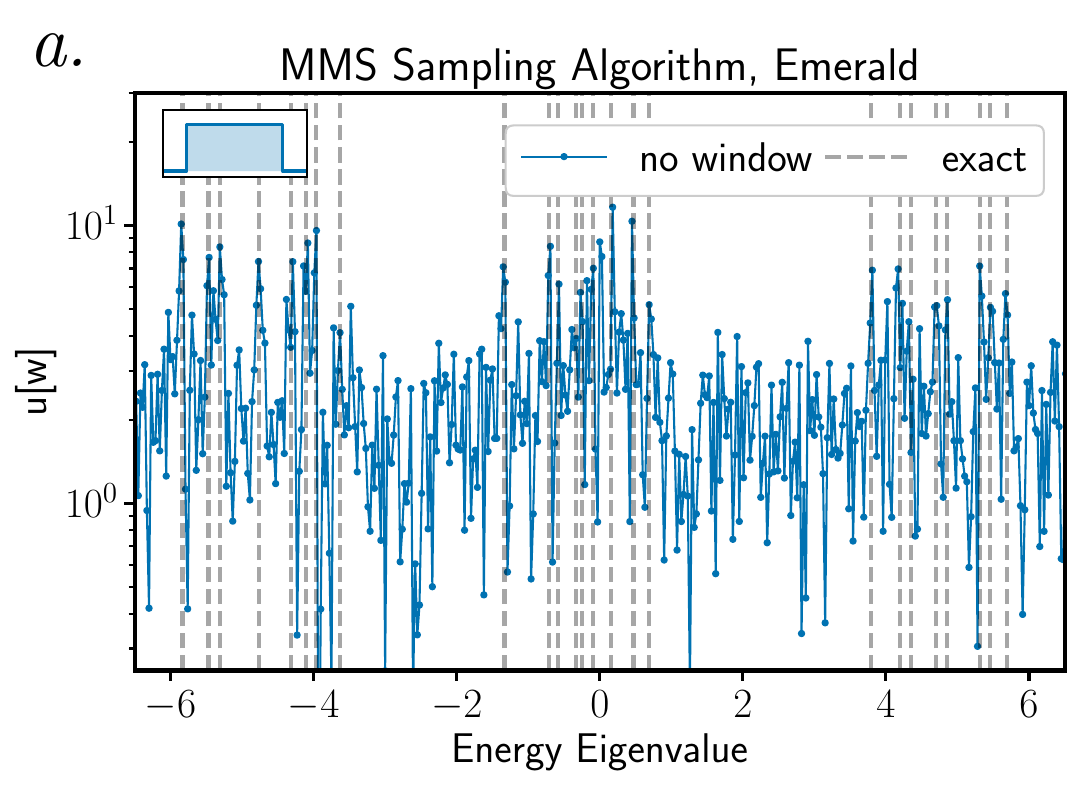}
\includegraphics[width=0.49\textwidth]{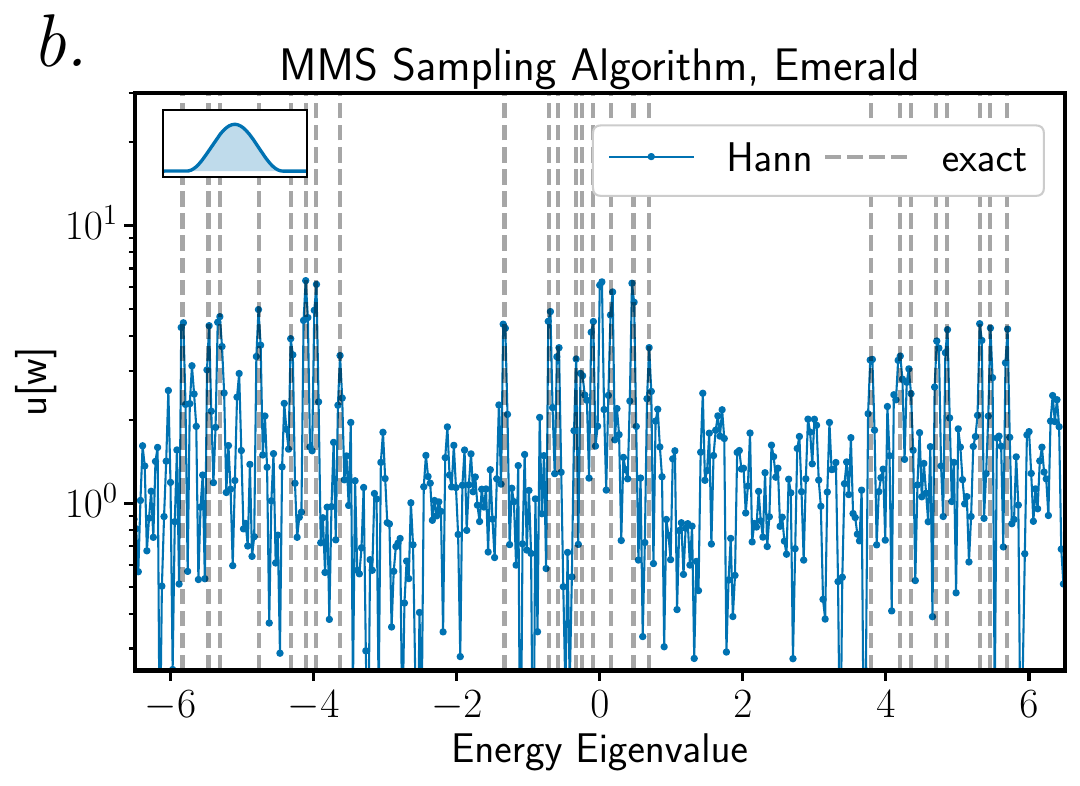}
\caption{The comparison of the data after the Fourier transform (\textit{a}) with no Hann window (rectangular weighting) and (\textit{b}) with a Hann window (raised cosine weighting). The data with the Hann window has broader peaks (widened by a factor of 2), but deeper troughs, making the peaks more obvious.}
\label{window_compare}
\end{figure}

These runs were performed on IQM Emerald over the course of two different calibration sessions. This section contains information about tests that were performed during each session before the runs began, and postprocessing techniques that were performed on the raw outputs from the runs.

At the beginning of each session, the best two-qubit gates are ranked by their fidelity scores. The layout is then chosen so that the controlled-evolution ancilla qubit is connected to four other qubits, each one of the Hamiltonian evolution qubits. Then the next best qubit (according to two-qubit fidelities) that is connected to the other Hamiltonian evolution qubits is chosen to be the fifth Hamiltonian evolution qubits. Finally, five more qubits are selected to act as garbage qubits, each connected to a Hamiltonian evolution qubit.

Based on this layout, the circuit is formed in a way that includes two error mitigation techniques: readout error mitigation and layer-based dynamical decoupling. Because measurements are based only on the readout of the single ancilla qubit which controls time-evolution, readout error mitigation is cheap and based on a $2 \times 2$ confusion matrix formed from the reported calibration for the session. The layer-based dynamical decoupling is applied only to the controlled evolution ancilla qubit and the Hamiltonian evolution qubits, not the garbage qubits. Dynamical decoupling acts every time in the circuit that a qubit is excluded from at least two consecutive two-qubit gates (single-qubit gates do not count in this). For this \textit{idle} qubit, an $X$ gate is applied twice, once $\approx 25\%$ of the time into the idle time, and once $\approx 75\%$ of the time into the idle time (according to the circuit layers), in order to fight the drift of decoherence. The code to produce the dynamical decoupling for these circuits is available on Zenodo.

In postprocessing, we present the Hann window results in the main text because it creates a better contrast in magnitude between the peaks versus the troughs, at the cost of widening the peaks. The Hann obtained is used by multiplying each data point of the trace of the time evolution by
\begin{equation}
w[k] = \tfrac{1}{2}\left(1 - \cos\frac{2\pi k}{N-1}\right),
\end{equation}
where $N$ is the number of timesteps and $k$ is the index of each data point, and then performing the FFT on that data. Figure \ref{window_compare} shows the data for both the standard rectangular window and the Hann window, and illustrates how the Hann window both lowers troughs but also broadens the peaks to twice the width of the rectangular window. For the current resolution, this is an acceptable tradeoff.

\section{Classical Methods for Estimating Densities of States}\label{sec:dos_methods}

The density of states provides a compact description of how the configurations or eigenstates of a gauge theory are distributed over an extensive variable. In the Hamiltonian formulation considered in this work, the density of states $\varrho_H(E)$ defined in Eq.~\eqref{eq:dos} counts energy eigenstates and determines the partition function through Eq.~\eqref{eq:Zfromdos}. From this one may calculate the free energy, entropy, internal energy, and specific heat over a continuous range of temperatures, quantities whose temperature-dependence may also imply phase transitions. 

A related construction may be defined in Euclidean lattice gauge theory. For a theory with fields $\phi$ and statistical weight $\exp[-\beta S[\phi]]$, one defines
\begin{equation}
    \rho(E)
    =
    \int \mathcal{D}\phi\,
    \delta\!\left(S[\phi]-E\right),
    \label{eq:euclidean_dos}
\end{equation}
so that
\begin{equation}
    Z(\beta)
    =
    \int dE\,\rho(E)e^{-\beta E}.
    \label{eq:euclidean_Z_dos}
\end{equation}
Here $E$ labels the value of the Euclidean action rather than a Hamiltonian energy eigenvalue, so this density of states should not be confused with $\varrho_H(E)$ even though both reduce the path integral to a single-variable integral.

Once the probability distribution
\begin{equation}
    P_\beta(E)
    =
    \frac{1}{Z(\beta)}\rho(E)e^{-\beta E}
\end{equation}
is computed, one may examine it for coexisting phases, suppressed regions between peaks, and free-energy barriers near first-order transitions. More generally, free-energy differences can be estimated from ratios between different regions of $\rho(E)$, or equivalently from the logarithm of their relative probabilities.

We note here that conventional importance sampling is designed to explore configurations which are typical at one chosen value of $\beta$. Such methods can consequently perform poorly when relevant configurations are rare, when metastable phases are separated by large barriers, or when the desired quantity depends on the normalization of the partition function rather than only on expectation values within a typical ensemble. Density of states methods instead quantify the relative abundance of configurations across an extended interval of $E$. In theories with a complex measure, this reformulation can improve the overlap problem and restrict the oscillatory contribution to a lower-dimensional integral, although it need not remove the cancellations responsible for the sign problem.

The Wang-Landau algorithm estimates the density of states by replacing canonical sampling with a random walk in energy \cite{Xu2006,Sinha2009}. By random walk, we mean that if the current estimate is $\rho_{\mathrm{est}}(E)$, configurations are sampled with a weight proportional to $1/\rho_{\mathrm{est}}(E)$, so that energies which have so far been assigned a small density are visited more often. In practice, one stores $g(E)=\ln \rho_{\mathrm{est}}(E)$ and updates bins iteratively with a controllable update size $\ln f$, while a histogram of visits is accumulated.

Once the histogram satisfies a chosen flatness criterion, it is reset and the update size is gradually reduced. The algorithm terminates when $\ln f$ is below a prescribed tolerance, at which point the random walk should have approximately compensated for the underlying variation of the density of states. This allows a single simulation to explore a broad energy range, including regions which are rarely visited in a canonical ensemble because their Boltzmann weight is exponentially smaller than that of the thermodynamically dominant energies. Wang-Landau sampling has also been applied to models with discrete energy levels \cite{Sinha2009,Langfeld2016}.

For systems with continuous energies, the density of states must be represented using finite energy bins. This situation arises, for example, for particles with unbound motion, scattering continua and continuum quantum field theories; systems with discrete finite-volume spectra develop an effectively continuous density of states in the thermodynamic limit. Within each bin, all configurations are assigned to the same estimated value of the density of states, even when $\rho(E)$ varies significantly across the interval \cite{Xu2006,Sinha2009}. The problem is most severe near the spectral boundaries, where because the density may change rapidly, the relevant configurations may be visited only rarely. The estimates density of states then becomes very sensitive to both the bin width and the sampling history \cite{Xu2006,Langfeld2012}.

Narrower bins reduce the bias from averaging over a finite interval, but they also increase the number of bins which must be sampled and the time required to visit each of them sufficiently often. This is the natural accuracy-cost tradeoff associated with Wang-Landau calculations \cite{Sinha2009}.

The LLR method extends global Wang-Landau sampling with a local determination of the logarithmic derivative of the density of states \cite{Langfeld2012,Langfeld2016}. The energy range is divided into intervals $[E_k,E_k+\delta E]$, and within each interval one approximates
\begin{equation}
    \ln \rho(E)
    \simeq
    \ln \rho\!\left(E_k+\frac{\delta E}{2}\right)
    +
    a_k
    \left(
        E-E_k-\frac{\delta E}{2}
    \right),
\end{equation}
where
\begin{equation}
    a_k
    =
    \left.
    \frac{d\ln\rho(E)}{dE}
    \right|_{E=E_k+\delta E/2}.
\end{equation}
In effect this method fits $\ln \rho(E)$ with a piecewise linear function. To determine the piecewise gradients $a_k$, LLR introduces a restricted and reweighted expectation value,
\begin{equation}
    \langle\!\langle f(E)\rangle\!\rangle_{k,a}
    =
    \frac{
        \displaystyle
        \int_{E_k}^{E_k+\delta E}
        dE\,\rho(E)f(E)e^{-aE}
    }{
        \displaystyle
        \int_{E_k}^{E_k+\delta E}
        dE\,\rho(E)e^{-aE}
    }.
\end{equation}
When the trial parameter $a$ equals the local slope $a_k$, the reweighting should have approximately accounated for the variation of $\rho(E)$ within the interval. The corresponding mean energy lies at the midpoint,
\begin{equation}
    \left\langle\!\left\langle
    E-E_k-\frac{\delta E}{2}
    \right\rangle\!\right\rangle_{k,a_k}
    =0.
\end{equation}
Restricted Monte Carlo sampling is a well-known method for solving this condition, with an Robbins-Monro update employed for stability in the presence of statistical noise \cite{Langfeld2016}.

Once the slopes have been determined in enough piecewise intervals, the global density of states may be reconstructed by integrating over these intervals. For example, the logarithm of the density at adjacent interval centers satisfies
\begin{equation}
    \ln
    \frac{
        \rho\!\left(E_{k+1}+\delta E/2\right)
    }{
        \rho\!\left(E_k+\delta E/2\right)
    }
    \simeq
    \frac{\delta E}{2}
    \left(a_k+a_{k+1}\right),
\end{equation}
so that repeated application determines $\ln\rho(E)$ up to an additive constant, or equivalently $\rho(E)$ up to an overall normalization.

The precision of the LLR estimate is controlled by the piecewise interval width $\delta E$. For the illustrative Gaussian density
\begin{equation}
    \rho(E)\propto
    \exp\!\left[-\frac{(E-\mu)^2}{2\sigma^2}\right],
\end{equation}
consider an interval whose center is $\bar E_k$. Since
\begin{equation}
    \ln\rho(E)
    =
    \ln\rho(\bar E_k)
    +
    a_k(E-\bar E_k)
    -
    \frac{(E-\bar E_k)^2}{2\sigma^2},
\end{equation}
the LLR linearization omits only the final quadratic term. Its largest relative deviation within the interval occurs at the endpoints and is
\begin{equation}
    \frac{\widetilde{\rho}(E)}{\rho(E)}-1
    =
    \exp\!\left(\frac{\delta E^2}{8\sigma^2}\right)-1
    =
    \frac{\delta E^2}{8\sigma^2}
    +O\!\left(\frac{\delta E^4}{\sigma^4}\right).
\end{equation}
which makes the expected $O(\delta E^2)$ discretization bias explicit \cite{Langfeld2016}. This example is chosen because the logarithmic slope is itself linear in $E$, so integrating over the LLR intervals reconstructs $\ln\rho$ unbiasedly at the interval centers; the remaining systematic error lies between those centers, where statistical errors accumulate and are indeed correlated across the reconstruction.

%

\end{document}